\newif\ifreport\reporttrue
\documentclass[sigconf]{acmart}
\def\BibTeX{{\rm B\kern-.05em{\sc i\kern-.025em b}\kern-.08emT\kern-.1667em\lower.7ex\hbox{E}\kern-.125emX}}
\renewcommand\footnotetextcopyrightpermission[1]{}
\setcopyright{none}
\usepackage{booktabs}
\usepackage{enumerate}
\usepackage{multirow}
\usepackage{bbm}
\usepackage{balance}
\usepackage{footnote}
\usepackage{tablefootnote}
\usepackage{xr}
\usepackage{xcolor}
\usepackage{subfigure}
\usepackage[lined, boxed, linesnumbered, ruled]{algorithm2e}
 
\theoremstyle{acmdefinition}

\theoremstyle{acmdefinition}\newtheorem{remark}{Remark}
\begin{document}

\title[]{Optimizing Information Freshness using Low-Power Status Updates via Sleep-Wake Scheduling}  %

\author{Ahmed M. Bedewy}
\orcid{1234-5678-9012}
\affiliation{%
  \institution{Department of ECE\\The Ohio State University}
  \city{Columbus}
  \state{OH}
}
\email{bedewy.2@osu.edu}

\author{Yin Sun}
\affiliation{%
  \institution{Department of ECE\\Auburn University}
  \city{Auburn}
  \state{AL}
}
\email{yzs0078@auburn.edu}

\author{Rahul Singh}
\affiliation{%
  \institution{Department of ECE\\The Ohio State University}
  \city{Columbus}
  \state{OH}
}
\email{singh.1434@osu.edu}

\author{Ness B. Shroff}
\affiliation{%
  \institution{Departments of ECE and CSE\\The Ohio State University}
   \city{Columbus}
  \state{OH}}
\email{shroff.11@osu.edu}

\begin{abstract}

In this paper, we consider the problem of optimizing the freshness of status updates that are sent from a large number of  low-power source nodes to a common access point. The source nodes utilize carrier sensing to reduce collisions and adopt an asychronized sleep-wake strategy to achieve an extended battery lifetime (e.g., 10-25 years). We use \emph{age of information} (AoI) to measure the freshness of status updates, and design the sleep-wake parameters for minimizing the weighted-sum peak AoI of the sources, subject to per-source battery lifetime constraints. When the sensing time is zero, this sleep-wake design problem can be solved by resorting to nested convex optimization; however, for positive sensing times, the problem is non-convex. We devise a low-complexity solution to solve this problem and prove that, for practical sensing times, the solution is within a small gap from the optimum AoI performance. Our numerical and NS-3 simulation results show that our solution can indeed elongate the batteries lifetime of information sources, while providing a competitive AoI performance.


\end{abstract}


\maketitle




\section{Introduction}\label{Int}
In applications such as networked monitoring and control systems, wireless sensor networks, autonomous vehicles, it is crucial for the destination node to receive timely status updates so that it can make accurate decisions.  \emph{Age of information} (AoI) has been used to measure the freshness of status updates. 
More specifically, AoI \cite{KaulYatesGruteser-Infocom2012} is the age of the freshest update at the destination, i.e., it is the time elapsed since the most recently received update was generated. It must be noted that optimizing traditional network performance metrics such as throughput or delay do not attain the goal of timely updating. For instance, it is well known that AoI could become very large when the offered load is high or low \cite{KaulYatesGruteser-Infocom2012}. 

In a variety of information update systems, energy consumption is also a critical constraint. For example, wireless sensor networks are used for monitoring crucial natural and human-related activities, e.g. forest fires, earthquakes, tsunamis, etc. Since such applications often require the deployment of sensor nodes in remote or hard-to-reach areas, they need to be able to operate unattended for long durations. Likewise, in medical sensor networks, since battery replacement/recharging involves a series of medical procedures, thereby providing disutility to patients, energy consumption must be constrained in order to support a long battery life of up to 10-15 years~\cite{timmons2004analysis}. 
Therefore, for networks serving such real-time applications, prolonging battery-life is just as crucial as guaranteeing a small AoI. Existing works on multi-source networks, e.g., \cite{yates2017status,talak2018distributed,li2013throughput,kadota2016minimizing_journal,hsu2017scheduling_2,jiang2018timely,kadota2018optimizing,talak2018optimizing2,aphermedis_he2017optimal,DBLP:journals/ton/GuoSKN18},  focused exclusively on minimizing the AoI and overlooked the need to reduce power consumption. This motivates us to derive algorithms that achieve a trade-off between the competing tasks of minimizing AoI and reducing the energy consumption in multi-source networks.

Additionally, some applications are characterized by a large number (typically hundreds of thousands) of densely packed wireless nodes serviced by only a single access point (AP). Examples include machine-type communication~\cite{kowshik2019energy}. The dataloads in such ``dense networks'' \cite{kowshik2019energy,kowshik2019fundamental} are created by applications such as home security and automation, oilfield and pipeline monitoring, smart agriculture, animal tracking and livestock, etc. This introduces high variability in the data packet sizes so that the transmission times of data packets are random. Thus scheduling algorithms that are designed for time-slotted systems with a fixed transmission duration, are not applicable to these systems. Besides that, synchronized scheduler for time-slotted systems are feasible when there are relatively few sources and each source has sufficient energy. However, if there are a huge number of sources, and each source has limited energy and low traffic rate, coordinating synchronized transmissions is quite challenging. This motivates us to design randomized protocols that coordinate the transmissions of multiple conflicting transmitters connected to a single AP.


Towards that end, we consider a wireless network with $M$ sources that contend for channel access and communicate their update packets to an AP. Each source is equipped with a battery that may get charged by a renewable source of energy, e.g., solar. Moreover, each source employs a ``sleep-wake'' scheme \cite{chen2013life} under which it transmits a packet if the channel is sensed idle; and sleeps if either: (i) It senses the channel to be busy, (ii) it completes a packet transmission. This enables each source to save the precious battery power by switching off at times when it is unlikely to gain channel access for packet transmissions. 

However, since a source cannot transmit during the sleep period, this causes the AoI to increase. We thus carefully design these sleeping periods so that the cumulative weighted average peak age of all sources is minimized, while ensuring that the energy consumption of each source is below its average battery power. 
To the best of our knowledge, this is the first work that considers AoI minimization in multi-source networks while simultaneously incorporating per-source battery lifetime constraints.

\subsection{Related Works}
There has been a significant effort on analyzing the AoI performance of popular queueing service disciplines, e.g., the First-Come, First-Served (FCFS) \cite{KaulYatesGruteser-Infocom2012}
 Last-Come, First-Served (LCFS) with and without preemption \cite{RYatesTIT16_2},
and queueing systems with packet management \cite{CostaCodreanuEphremides_TIT}.
 In \cite{age_optimality_multi_server,Bedewy_NBU_journal_2,multihop_optimal,bedewy2017age_multihop_journal_2,Yin_multiple_flows}, the age-optimality of Last-Generated, First-Served (LGFS)-type policies in multi-server and multi-hop networks was established, where it was shown that these policies can minimize any non-decreasing functional of the age processes. 
The fundamental coupling of data sampling and transmission in information update systems was investigated in \cite{SunJournal2016,sun2018sampling_2}, where sampling policies were designed to minimize any nonlinear age functions in single source systems. These studies were later extended to a multi-source scenario in \cite{multi_source_bedewy_2}. 

Designing scheduling policies for minimizing AoI in multi-source networks has recently received increasing attention, e.g., \cite{DBLP:journals/ton/GuoSKN18,yates2017status,talak2018distributed,li2013throughput,kadota2016minimizing_journal,hsu2017scheduling_2,jiang2018timely,kadota2018optimizing,talak2018optimizing2,aphermedis_he2017optimal}. Of particular interest, are those pertaining to designing distributed scheduling policies \cite{yates2017status,talak2018distributed,li2013throughput,hsu2017scheduling_2,kadota2016minimizing_journal,jiang2018timely}. The work in \cite{yates2017status} considered slotted ALOHA-like random access scheme in which each node accesses the channel with a certain access probability. These probabilities were then optimized in order to minimize the AoI. However, the model of~\cite{yates2017status} allows multiple interfering users to gain channel access simultaneously, and hence allows for the collision. The authors in \cite{talak2018distributed} generalized the work in~\cite{yates2017status} to a wireless network in which the interference is described by a general interference model. 
The Round Robin or Maximum Age First policy was shown to be (near) age-optimal for different system models, e.g., in \cite{li2013throughput,hsu2017scheduling_2,kadota2016minimizing_journal,jiang2018timely}.

A central component of the scheme proposed in this work is the carrier sensing mechanism in which sources sense the channel to detect times during which no interfering transmissions occur. We note that such mechanisms are employed in numerous distributed medium-access schemes in wireless networks, such as Carrier Sense Multiple Access (CSMA), see~\cite{yun2012optimal} for a recent survey of the existing schemes. Thus, there has been an interest in designing CSMA-based scheduling schemes that optimize the AoI \cite{maatouk2019minimizing,wang2019broadcast}. In \cite{maatouk2019minimizing}, the authors employed the standard idealized CSMA in \cite{jiang2010distributed} to minimize the AoI with an exponentially distributed packet transmission times.  In \cite{wang2019broadcast}, the authors employed the slotted Carrier Sense Multiple Access/Collision-Avoidance (CSMA/CA) in \cite{bianchi2000performance} to minimize the broadcast age of information, which is defined, from a sender perspective, as the age of the freshest successfully broadcasted packet. Contrary to these works, the sleep-wake scheme proposed by us emphasizes on reducing the cumulative energy consumption in multi-source networks in addition to minimizing the cumulative weighted AoI. Moreover, in our study, transmission times are not necessarily random variables with some commonly used parametric density \cite{maatouk2019minimizing}, or deterministic \cite{wang2019broadcast}, but can be any generally distributed random variables with finite mean.


\begin{figure*}[t]
\includegraphics[scale=.6]{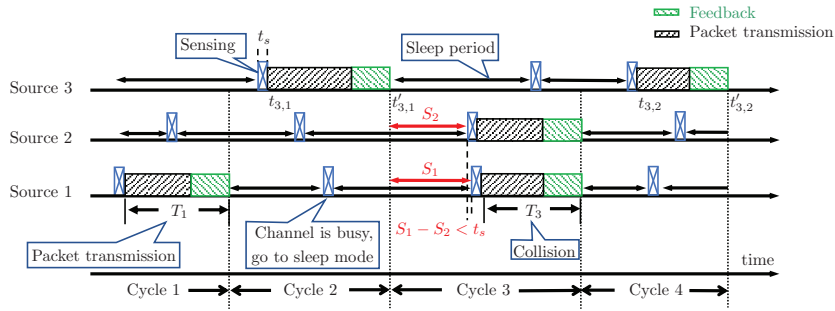}
\centering
\captionsetup{justification=justified}
\caption{Illustration of the sleep-wake cycles. In Cycle 1-2, we have successful packet transmissions. Let $S_1$ and $S_2$ represent the remaining sleeping times of Sources 1 and 2,  respectively, after a successful transmission. Then, a collision occurs in Cycle 3 because 
 the difference between wake-up times of Sources 1 and 2 is less than $t_s$, i.e., $S_1-S_2<t_s$. As we can observe, each cycle consists of an idle period before a transmission/collision event.}
 \label{sleep_wake_cycle}
\end{figure*}
\subsection{Key Contributions}
Our key contributions are summarized as follows:
\begin{itemize}

\item The problem of minimizing the total weighted average peak age over the sources, while simultaneously meeting per-source energy constraints is non-convex. Nonetheless, we devise a solution, i.e., a choice of the mean sleeping durations for each source. We then show that in the regime for which the sensing time is negligible compared to the packet transmission time, the proposed solution is near-optimal (Theorem \ref{thm_case_sum_bi_ge_1} and Theorem \ref{thm_case_sum_bi_le_1}). Our near-optimality results hold for any generally distributed packet transmission times.

\item  We propose an algorithm that can be easily implemented in many industrial control systems. In particular, we are able to represent our solution in a form that requires the knowledge of two universal parameters to obtain its value. These universal parameters are functions of network parameters, i.e., the mean packet transmission times, carrier sensing time, energy constraint information, weight of each source. Hence, the proposed algorithm requires the source nodes to share the network parameters with an AP that is connected to all the sources. Once the AP obtains this information, it calculates these universal parameters and broadcasts them to all the sources. Each source, thereafter, uses these universal parameters to compute its mean sleeping times.

\item Finally, in the limiting scenario, when the ratio between the sensing time and the packet transmission time goes to zero, we show that the age performance of our proposed algorithm is as good as that of the optimal synchronized scheduler (e.g., for time-slotted systems), in which the time overhead needed for coordinating different sources with random packet sizes are omitted (Corollary \ref{corollary_centeralized_scheduler}).  



\end{itemize}

\section{Model and Formulation}\label{sysmod}
\subsection{Network Model and Sleep-wake Scheduling}
Consider a wireless network composed of $M$ source nodes observing time-varying processes. Sources generate update packets and communicate them to an access point (AP) over the same spectrum band. If multiple sources transmit packets simultaneously, a packet collision occurs and the corresponding packet transmissions fail. 


We assume that the sources use a sleep-wake scheduling scheme to access the shared channel, where the sources switch between a sleep mode and transmission mode over time, according the following rules: Upon waking from the sleep mode, a source first performs carrier sensing to check whether the channel is occupied by another source, as illustrated in Figure \ref{sleep_wake_cycle}. We assume that the sources are within the hearing range of each other. The time duration of carrier sensing is denoted as $t_s$, which is sufficiently long to ensure a high sensing accuracy.
If the channel is sensed to be busy, the source enters the sleep mode directly; otherwise,  the source generates and transmits an update packet over the channel. Upon completing a packet transmission, the source goes back to the sleep mode. 

In the above sleep-wake scheduling scheme, if two sources start transmitting within a duration of $t_s$, then they may not be able to sense the transmission of each other. In order to obtain a robust system design, we consider that they cannot detect each other's transmission in this case and a collision occurs. A feedback is sent back to the sources to indicate the outcome of their transmissions (successful transmission or collision).

A \emph{sleep-wake} cycle, or simply a \emph{cycle}, is defined as the time period between the ends of two successive packet transmission or collision events in the network. Each cycle consists of an idle period before a transmission/collision event. As depicted in Figure \ref{sleep_wake_cycle}, the packet transmissions in Cycle 1-2 are successful, but a collision occurs in Cycle 3 because Sources 1 and 2 wake up within a short duration $t_s$.   

 We use $T_i,i\in \{1,2,\ldots\}$ to represent the time incurred during the $i$-th packet transmission or collision event over time, which includes propagation and feedback delays. For example, in Figure~\ref{sleep_wake_cycle}, $T_1$ is the duration of the packet transmission event by Source~1, while $T_3$ is the duration of the collision event between Source~1 and 2. We assume that the distribution of the time spent during transmission or collision is the same.   In Section \ref{ns3sim}, we show that this assumption has a negligible impact on the performance of the proposed algorithm. The transmission/collision times $T_i$'s are i.i.d. across time and sources, and are generally distributed. In the rest of the paper, we omit the subscript $i$ of $T_i$ for simplicity, and use $T$ to denote the transmission/collision time, which is assumed to have a finite mean, i.e., $E[T]<\infty$. The sleep periods of source~$l$ are exponentially distributed random variables with mean value $\mathbb{E}[T]/r_l$ and are independent across sources and \emph{i.i.d.} across time. Here, the sleep period parameter $r_l$ has been normalized by the mean transmission time $\mathbb{E}[T]$. Let $\mathbf{r}=(r_1, \ldots, r_M)$ be the vector comprising of these sleep period parameters.

\subsection{Total Weighted Average Peak Age }\label{objective_function1}
Let $\alpha_l$ be the probability of the event that the source~$l$ obtains channel access and successfully transmits a packet within a cycle. It follows from \cite{chen2013life} that $\alpha_l$ is given by
\begin{align}\label{access_prob_in_agiven_cycle}
\alpha_l=\frac{r_l e^{r_l\frac{t_s}{\mathbb{E}[T]}}}{e^{\sum_{i=1}^Mr_i\frac{t_s}{\mathbb{E}[T]}}\sum_{i=1}^Mr_i}.
\end{align}
\ifreport In order to keep the discussion self-contained, we derive the above expression in  Appendix \ref{Appendix_A'} \fi. Let $N_l$ denote the total number of cycles between two successful transmissions of source~$l$. Now, if the probability that source~$l$ obtains channel access and transmits successfully in a given cycle is $\alpha_l$, and $1-\alpha_l$ otherwise, then $N_l$ is geometrically distributed with mean $\frac{1}{\alpha_l}$. Thus, we get
\begin{equation}\label{mean_nf}
\mathbb{E}[N_l]=\frac{e^{\sum_{i=1}^Mr_i\frac{t_s}{\mathbb{E}[T]}}\sum_{i=1}^Mr_i}{r_l e^{r_l\frac{t_s}{\mathbb{E}[T]}}}.
\end{equation}

\begin{figure}[h]
\includegraphics[scale=0.22]{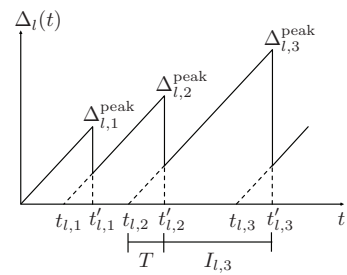}
\centering
\captionsetup{justification=justified}
\caption{The age evolution of source~$l$ ($\Delta_l(t)$). 
}
\label{age_proc}
\end{figure}
 Let $U_l(t)$ represent the generation time of the most recently delivered packet from source~$l$ by time $t$. Then, the \emph{age of information}, or simply the \emph{age}, of source~$l$ is defined as \cite{KaulYatesGruteser-Infocom2012} 
\begin{equation}
\Delta_l(t)=t-U_l(t). 
\end{equation}
As shown in Figure \ref{age_proc}, the age increases linearly with $t$, but is reset to a smaller value upon the delivery of a fresher packet. Since a fresh update packet is delivered each time a source obtains channel access and completes transmission, the AoI of source~$l$ is reset after a random number of $N_l$ cycles. We suppose that the age $\Delta_l(t)$ is right-continuous. 
\ifreport 
Observe that a small age $\Delta_l(t)$ indicates that the AP has a fresh status update packet that was generated at source~$l$ recently. Hence, it is desirable to keep $\Delta_l(t)$ small for all the sources.
\fi

We begin by introducing some notations and definitions. We use $t_{l,i}$ and $t_{l,i}'$ to denote the generation and delivery times, respectively, of the $i$-th delivered packet from source~$l$, where we have $t_{l,i}'-t_{l,i}=T$.\footnote{A packet of a particular source is deemed delivered when the source receives the feedback.}  Let $I_{l,i}=t_{l,i}'-t_{l,i-1}'$ denote the $i$-th inter-departure time of source~$l$, where we have $\mathbb{E}[I_{l,i}]=\mathbb{E}[I_l]\leq \infty$ for all $i$. 
The $i$-th peak age of source~$l$, denoted by $\Delta^{\text{peak}}_{l,i}$, is defined as the AoI of source~$l$ right before the $i$-th packet delivery from source~$l$, i.e., we have 
\begin{align}
\Delta^{\text{peak}}_{l,i}=\Delta_l(t_{l,i}'^-),
\end{align}
where $t_{l,i}'^-$ is the time instant just before the delivery time $t_{l,i}'$. This is shown in Figure \ref{age_proc}. The average peak age metric provides information regarding the worst case age with the advantage of having a simpler formulation than the average age metric \cite{CostaCodreanuEphremides_TIT}. Thus, it is suitable for applications that have an upper bound restriction on AoI. One can observe from Figure \ref{age_proc}, that the peak age can be expressed as \cite{CostaCodreanuEphremides_TIT}
\begin{align}
\Delta^{\text{peak}}_{l,i}=I_{l,i}+T.
\end{align}
Hence, the average peak age of source~$l$ is given by 
 \begin{equation}\label{peak_age1}
\mathbb{E}[\Delta^{\text{peak}}_{l,i}]=\mathbb{E}[I_l]+\mathbb{E}[T].
\end{equation}
We now derive an expression for $\mathbb{E}[I_l]$. An inter-departure time duration of a particular source is composed of multiple consecutive sleep-wake cycles, see Figure \ref{sleep_wake_cycle}. With a slight abuse of notation, we let $\textbf{cycle}_{l,i}$ denote the duration of the $i$-th sleep-wake cycle after a successful transmission of source~$l$. Hence, we have
\begin{equation}
\mathbb{E}[I_l]=\mathbb{E}\left[\sum_{i=1}^{N_l}\textbf{cycle}_{l,i}\right].
\end{equation}
Note that $\textbf{cycle}_{l,i}$'s are i.i.d. across time. Moreover, since $\mathbb{P}(N_l=n)$ depends only on the history, $N_l$ is a stopping time \cite{shiryaev2007optimal}. Hence, it follows from Wald's identity \cite{wald1973sequential} that
\begin{equation}\label{inter_dep}
\mathbb{E}[I_l]=\mathbb{E}[N_l]\mathbb{E}[\textbf{cycle}],
\end{equation}
where $\mathbb{E}[\textbf{cycle}]$ is the mean duration of a sleep-wake cycle. Each cycle consists of an idle period and a transmission/collision time, see Figure \ref{sleep_wake_cycle}. Using the memoryless property of exponential distribution, we observe that the idle period is the minimum of exponential random variables. Thus, it can be shown that the idle period in each cycle is exponentially distributed with mean value equal to $\mathbb{E}[T]/\sum_{i=1}^Mr_i$, where $\mathbb{E}[T]/r_l$ is the mean of sleep periods of source~$l$. Hence, we have
\begin{equation}\label{swc_mean}
\mathbb{E}[\textbf{cycle}]=\frac{\mathbb{E}[T]}{\sum_{i=1}^Mr_i}+\mathbb{E}[T].
\end{equation}
Substituting the expressions for $\mathbb{E}[N_l]$ and $\mathbb{E}[\textbf{cycle}]$  from \eqref{mean_nf} and  \eqref{swc_mean}, respectively, into \eqref{inter_dep}, and then into \eqref{peak_age1}, we obtain
\begin{equation}
\begin{split}
\mathbb{E}[\Delta^{\text{peak}}_{l,i}]=&\frac{e^{-r_l\frac{t_s}{\mathbb{E}[T]}}\mathbb{E}[T]}{r_l}e^{\sum_{i=1}^Mr_i\frac{t_s}{\mathbb{E}[T]}}\left(1+\sum_{i=1}^Mr_i\right)+\mathbb{E}[T].
\end{split}
\end{equation}
In this paper, we aim to minimize the total weighted average peak age, which is given by
\begin{equation}\label{t_avg_peak_age}
\begin{split}
\!\!\sum_{l=1}^M\frac{w_le^{-r_l\frac{t_s}{\mathbb{E}[T]}}\mathbb{E}[T]}{r_l}e^{\sum_{i=1}^Mr_i\frac{t_s}{\mathbb{E}[T]}}\left(1\!+\!\sum_{i=1}^Mr_i\right)+\sum_{l=1}^Mw_l\mathbb{E}[T],
\end{split}
\end{equation}
where $w_l>0$ is the weight of source~$l$. The weights here enable us to prioritize the sources according to their relative importance \cite{talak2018distributed,talak2018optimizing2}.

\subsection{Energy Constraint}
Each source is equipped with a battery that can possibly be recharged by a renewable energy source, such as solar. The  energy constraint on source~$l$ is described by the following parameters:
a) Initial battery level $B_{l}$, which denotes the initial amount of energy stored in its battery, b) Target lifetime $D_{l}$, which is the minimum time-duration that the source~$l$ should be active before its battery is depleted, c) Average energy replenishment rate\footnote{It is assumed that $R_l$ is either known, or it can be estimated accurately.} $R_l$, which is the rate at which the battery of source~$l$ receives energy from its energy source. Observe that if source~$l$ does not have access to an energy source, then we have $R_l=0$.

In typical wireless sensor networks, sources have a much smaller power consumption in the sleep mode than in the transmission mode. For example, the power consumption in the sleep mode is 15 $\mu$W while the power consumption in the transmission mode is 24.75 mW \cite{sleep_radio_enrgy_cons}. Motivated by this, we assume that the energy dissipation during sleep modes is negligible as compared to the power consumption in the transmission mode. Moreover, we assume that the sensing time duration $t_s$ is very short as compared to the transmission time and hence neglect the energy consumed while sensing the channel. In Section \ref{ns3sim}, we show that these assumptions have a negligible effect on the performance of the proposed algorithm. Under these assumptions, the amount of energy used by a source is  equal to the amount of energy consumed in transmissions. Note that the power consumed in packet transmission is equal to the sum of energy consumed while using radio signal during packet transmission, and the power used for receiving feedback.
 
The maximum allowable energy consumption rate for transmissions, denoted by $E_{\text{con},l}$, is given by
\begin{equation}
E_{\text{con},l}=\frac{B_l}{D_l}+R_l, ~\forall l.
\end{equation}
Then,  for source $l$ to achieve its target lifetime, $D_l$, the actual energy consumption rate of source~$l$, $E_l$, must satisfy
\begin{equation}\label{energy_const_1}
E_l\leq E_{\text{con},l}, ~\forall l.
\end{equation}
For the sleep-wake mechanism under consideration, it has been shown in \cite{chen2013life} that the total fraction of time in which source~$l$ transmits update packets is given by
\begin{equation}\label{sigma_l}
\sigma_l=\frac{[1-e^{-r_l\frac{t_s}{\mathbb{E}[T]}}]\sum_{i=1}^Mr_i+r_le^{-r_l\frac{t_s}{\mathbb{E}[T]}}}{\sum_{i=1}^Mr_i+1}.
\end{equation}
For the sake of completeness, the derivation of $\sigma_l$ is discussed in \ifreport Appendix \ref{Appendix_A''} \fi. If $E_{\text{avg},l}$ is the average energy consumption rate of source~$l$ in the transmission mode, then we have
\begin{equation}
E_l=\sigma_lE_{\text{avg},l},~ \forall l.
\end{equation}
Define $b_l\triangleq E_{\text{con},l}/E_{\text{avg},l}$ as the target energy efficiency of source~$l$. Then, the energy constraints in \eqref{energy_const_1} can be rewritten as
\begin{equation}\label{energy_const_2}
\sigma_l=\frac{[1-e^{-r_l\frac{t_s}{\mathbb{E}[T]}}]\sum_{i=1}^Mr_i+r_le^{-r_l\frac{t_s}{\mathbb{E}[T]}}}{\sum_{i=1}^Mr_i+1}\leq b_l, ~\forall l.
\end{equation}
Observe that if $b_l\geq 1$, then constraint \eqref{energy_const_2} is always satisfied. 

\subsection{Problem Formulation}
Our goal is to design $\mathbf{r}$ in order to minimize the total weighted average peak age in \eqref{t_avg_peak_age}, while simultaneously ensuring that the energy constraints \eqref{energy_const_2} are satisfied. After normalizing the total weighted average peak age in \eqref{t_avg_peak_age} by $\mathbb{E}[T]$, our goal can be cast as the following optimization problem: 
(Problem \textbf{1})
\begin{align}\label{problem1}
\!\!\!\!\!\!\!\!\!\!\!\!\!\!\!\!\!\!\!\!\!\!\!\!\!\!\!\!\!\!\!\!\begin{split}
\bar{\Delta}^{\text{peak}}_{\text{opt}}\triangleq\min_{r_l>0}& \sum_{l=1}^M\frac{w_le^{-r_l\frac{t_s}{\mathbb{E}[T]}}}{r_l}e^{\sum_{i=1}^Mr_i\frac{t_s}{\mathbb{E}[T]}}\left(1+\sum_{i=1}^Mr_i\right)+\sum_{l=1}^Mw_l\\
\textbf{s.t.}~&\frac{[1-e^{-r_l\frac{t_s}{\mathbb{E}[T]}}]\sum_{i=1}^Mr_i+r_le^{-r_l\frac{t_s}{\mathbb{E}[T]}}}{\sum_{i=1}^Mr_i+1}\leq b_l,\forall l,
\end{split}\!\!\!\!\!\!\!\!\!\!\!\!\!\!\!\!\!\!\!\!\!\!\!\!
\end{align}
where $\bar{\Delta}^{\text{peak}}_{\text{opt}}$ is the optimal objective value of Problem \textbf{1}. We will use $\bar{\Delta}^{\text{peak}}(\mathbf{r})$ to denote the objective function of Problem \textbf{1} for given sleeping period parameters $\mathbf{r}$. One can notice from \eqref{problem1} that the optimal sleeping period parameters depends on the sensing time $t_s$ and the mean transmission time $\mathbb{E}[T]$ only through their ratio $t_s/\mathbb{E}[T]$. This insight plays a crucial role in subsequent analysis of Problem \textbf{1}.


\section{Main Results}\label{main_result}
We can observe that Problem \textbf{1} can be solved by resorting to nested convex optimization, if the sensing time is zero. However, 
Problem \textbf{1} becomes non-convex for positive sensing times. Hence, it is challenging to solve for optimal $\mathbf{r}$. In this section we will propose a low-complexity closed-form solution which is shown to be near-optimal when the sensing time is small as compared with the transmission time. Our solution is developed by considering the following two regimes separately: (i) \emph{Energy-adequate regime} denoted as $\sum_{i=1}^M b_i \geq 1$, where the condition $\sum_{i=1}^M b_i \geq 1$ means that the sources have a sufficient amount of total energy to ensure that at least one source is awake at any time, (ii) \emph{Energy-scarce regime} represented by $\sum_{i=1}^M b_i < 1$, which indicates that the sources have to sleep for some time to meet the sources' energy constraints.
\subsection{Energy-adequate Regime}
In the energy-adequate regime $\sum_{i=1}^M b_i \geq 1$, our solution $\mathbf{r}^{\star}:= (r^{\star}_1,\ldots,r^{\star}_M)$ is given as 
\begin{equation}\label{r_f_b_gr_1}
r^{\star}_l=\min\{b_l,\beta^{\star}\sqrt{w_l}\} x^{\star},\forall l,
\end{equation}
where $x^\star$ and $\beta^\star$ are expressed in terms of the the parameters $\{b_i,w_i\}_{i=1}^{M},t_s/\mathbb{E}[T] $
as follows:
\begin{equation}\label{x*}
x^{\star}=\frac{-1}{2}+\sqrt{\frac{1}{4}+\frac{\mathbb{E}[T]}{t_s}},
\end{equation}
and $\beta^{\star}$ is the root of 
\begin{equation}\label{condition_beta}
\sum_{i=1}^M\min\{b_i,\beta^{\star}\sqrt{w_i}\}\ =1.
\end{equation}
The performance of the above solution $\mathbf{r}^\star$ is manifested in the following theorem:
%
%
%
\begin{theorem}[\textbf{Near-optimality}]\label{thm_case_sum_bi_ge_1}
If $\sum_{i=1}^Mb_i\ge 1$, then the solution $\mathbf{r}^{\star}$ given by \eqref{r_f_b_gr_1} - \eqref{condition_beta} is near-optimal for solving \eqref{problem1} when $t_s/E[T]$ is sufficiently small, in the following sense:\footnote{We use the standard order notation: $f(h)=O(g(h))$ means $z_1\le \lim_{h\to 0}f(h)/g(h)\le z_2$ for some constants $z_1>0$ and $z_2>0$, while $f(h)=o(g(h))$ means $\lim_{h\to 0}f(h)/g(h)= 0$.}
\begin{align}\label{sub_opt_gap_eq}
\vert\bar{\Delta}^{\text{peak}}(\mathbf{r}^{\star})-\bar{\Delta}^{\text{peak}}_{\text{opt}}\vert \leq 2\sqrt{\frac{t_s}{\mathbb{E}[T]}}C_1\!+\!o\left(\sqrt{\frac{t_s}{\mathbb{E}[T]}}\right),
\end{align}
where
\begin{align}\label{eq23thm3_1}
C_1=\sum_{i=1}^M\frac{w_i}{\min\{b_i,\beta^{\star}\sqrt{w_i}\}}.
\end{align}
\end{theorem}
 \ifreport
\begin{proof}
See Section \ref{proof_case_sum_bi_ge_1}. 
\end{proof}
\fi
As a result of Theorem \ref{thm_case_sum_bi_ge_1}, we can obtain the following corollary:
\begin{corollary}[\textbf{Asymptotic optimality}]\label{cor1}
If $\sum_{i=1}^Mb_i\ge 1$, then the solution $\mathbf{r}^{\star}$ given by \eqref{r_f_b_gr_1} - \eqref{condition_beta} is asymptotically optimal for the Problem \textbf{1} as $t_s/\mathbb{E}[T]\rightarrow 0$, i.e., 
\begin{align}\label{zero_gab_result}
\lim_{\frac{t_s}{\mathbb{E}[T]}\rightarrow 0} \vert\bar{\Delta}^{\text{peak}}(\mathbf{r}^{\star})-\bar{\Delta}^{\text{peak}}_{\text{opt}}\vert=0.
\end{align}
Moreover, the asymptotic optimal value of Problem \textbf{1} as $t_s/\mathbb{E}[T]\to 0$ is
  \begin{equation}\label{asymptotic_value_final}
\lim_{\frac{t_s}{\mathbb{E}[T]}\rightarrow 0} \bar{\Delta}^{\text{peak}}_{\text{opt}}=\sum_{i=1}^M\left[\frac{w_i}{\min\{b_i,\beta^{\star}\sqrt{w_i}\}}+w_i\right].
\end{equation}
 \end{corollary}
 \ifreport
\begin{proof}
See Section \ref{proof_case_sum_bi_ge_1}. 
\end{proof}
\fi 
 \subsection{Energy-scarce Regime}\label{case_sum_bi_le_1}
Now, we present a solution to Problem \textbf{1} and show it is near-optimal in energy-scarce regime $\sum_{i=1}^Mb_i< 1$. The solution $\mathbf{r}^\star$ of the energy-scarce regime is again given by \eqref{r_f_b_gr_1}, where $x^\star$ and $\beta^\star$ are determined as
\begin{equation}\label{condition_x*_beta*}
x^{\star}=\frac{\min_l c_l}{1-\sum_{i=1}^Mb_i},~\beta^{\star}=\sum_{i=1}^M\frac{1}{\sqrt{w_i}},
\end{equation}
and 
\begin{align}
\!\!\!\!\!\!c_l =&\frac{2b_l\left(1-\sum_{i=1}^Mb_i\right)^2}{Q},\label{feasible_factor} \\
\begin{split}
\!\!\!\!\!\!Q= &b_l\left(\!1\!-\!\sum_{i=1}^Mb_i\!\right)^2\\
&\!\!\!\!\!\!+\!\sqrt{b_l^2\left(1\!-\!\sum_{i=1}^Mb_i\right)^4\!\!\!+4b_l^2\left(1\!-\!\sum_{i=1}^Mb_i\right)^2\!\left(\sum_{i=1}^Mb_i\!-\!b_l\right)\frac{t_s}{\mathbb{E}[T]}}.\label{feas_fact_eq27}
\end{split}
\end{align}
Then, the near-optimality of the proposed solution (i.e., $\mathbf{r}^{\star}$) is explained in the following theorem:
\begin{theorem}[\textbf{Near-optimality}]\label{thm_case_sum_bi_le_1}
If $\sum_{i=1}^Mb_i< 1$, then  the solution $\mathbf{r}^{\star}$ given by \eqref{r_f_b_gr_1} and  \eqref{condition_x*_beta*} - \eqref{feas_fact_eq27} is near-optimal for solving \eqref{problem1} when $t_s/\mathbb{E}[T]$ is sufficiently small, in the following sense:
\begin{align}\label{sub_opt_gap_eq_2}
\vert\bar{\Delta}^{\text{peak}}(\mathbf{r}^{\star})-\bar{\Delta}^{\text{peak}}_{\text{opt}}\vert \leq \frac{t_s}{\mathbb{E}[T]}C_2\!+\!o\left(\frac{t_s}{\mathbb{E}[T]}\right),
\end{align}
where
\begin{align}\label{eq29thm3_3}
C_2=\sum_{l=1}^M\frac{w_l}{b_l(1-\sum_{i=1}^Mb_i)}\left(3\sum_{i=1}^Mb_i-\min_j b_j\right).
\end{align}
\end{theorem}
 \ifreport
\begin{proof}
See Section \ref{proof_case_sum_bi_le_1}. 
\end{proof}
\fi
From Theorem \ref{thm_case_sum_bi_le_1}, we obtain the following corollary:
\begin{corollary}[\textbf{Asymptotic optimality}]\label{cor2}
If $\sum_{i=1}^Mb_i< 1$, then \eqref{zero_gab_result} holds for the solution $\mathbf{r}^{\star}$ given by \eqref{r_f_b_gr_1} and \eqref{condition_x*_beta*} - \eqref{feas_fact_eq27}. Hence, our proposed solution is asymptotically optimal for the Problem \textbf{1} as $t_s/\mathbb{E}[T]\rightarrow 0$. Moreover, the asymptotic optimal value of Problem \textbf{1} as $t_s/\mathbb{E}[T]\to 0$ is 
\begin{equation}\label{asymptotic_value_final_2}
\begin{split}
\lim_{\frac{t_s}{\mathbb{E}[T]}\rightarrow 0} \bar{\Delta}^{\text{peak}}_{\text{opt}}&=\sum_{i=1}^M\left[\frac{w_i}{\min\{b_i,\beta^{\star}\sqrt{w_i}\}}+w_i\right]\\&
=\sum_{i=1}^M\left[\frac{w_i}{b_i}+w_i\right].
\end{split}
\end{equation}
\end{corollary}
 \ifreport
\begin{proof}
See Section \ref{proof_case_sum_bi_le_1}. 
\end{proof}
\fi
Interestingly,  the asymptotic optimal values of Problem \textbf{1} in both regimes, given by \eqref{asymptotic_value_final} and \eqref{asymptotic_value_final_2}, are identical. However, in the energy-scarce regime, we can observe that $\beta^{\star}$, which is defined in \eqref{condition_x*_beta*}, always satisfies $\min\{b_l,\beta^{\star}\sqrt{w_l}\}=b_l$ for all $l$.
\begin{remark}
We would like to point out that the condition $t_s/\mathbb{E}[T]\approx 0$ is satisfied in many practical applications. For instance, in wireless sensor networks \cite{el2002spatial}, the carrier  sensing   time   is $t_s=40~\mu$s, while the transmission time is around $5$ ms. Hence, $t_s/\mathbb{E}[T]\approx 0.008$.
\end{remark}
\subsection{Discussion}
In this subsection, we discuss a simple implementation of our proposed solution. Moreover, we provide some useful insights about our proposed solution at the limit point $t_s/\mathbb{E}[T]\to 0$. 

\subsubsection{Implementation of Sleep-wake Scheduling}
We devise a simple algorithm to compute our solution $\mathbf{r}^\star$, which is provided in Algorithm \ref{alg1}. Notice that $\mathbf{r}^\star$ has the same expression \eqref{r_f_b_gr_1} in the energy-adequate and energy-scarce regimes. We exploit this fact to simplify the implementation of sleep-wake scheduling. In particular, the sources report $w_l$ and $b_l$ to the AP, which computes $\beta^\star$ and $x^\star$, and broadcasts them back to the sources. After receiving $\beta^\star$ and $x^\star$,  source~$l$ computes $r_l^\star$ based on \eqref{r_f_b_gr_1}. In practical wireless sensor networks, e.g., smart city networks and industrial control sensor networks \cite{wsn_industry,hsieh2018decentralized}, the sensors report their measurements via an access point (AP). Hence, it is reasonable to employ the AP in implementing the sleep-wake scheduler.
\begin{algorithm}[h]
\footnotesize
\SetKwData{NULL}{NULL}
\SetCommentSty{footnotesize} 
The AP gathers the parameters  $\{(w_i,b_i)_{i=1}^M,t_s/\mathbb{E}[T]\}$\;
\uIf{$\sum_{i=1}^Mb_i\ge 1$}{
The AP derives $x^{\star},\beta^{\star}$ according to \eqref{x*} and \eqref{condition_beta}\;}
\Else{
The AP derives  $x^{\star},\beta^{\star}$ according to \eqref{condition_x*_beta*} - \eqref{feas_fact_eq27}\;}
The AP broadcasts $x^{\star},\beta^{\star}$ to all the $M$ sources\;
Upon hearing $x^{\star},\beta^{\star}$, source~$l$ compute $r^{\star}_l$ from \eqref{r_f_b_gr_1}\;
\caption{Implementation of sleep-wake scheduler.}\label{alg1}
\end{algorithm}

In the above implementation procedure, the sources do need not know if the overall network is in the energy-adequate or energy-scarce regime; only the AP knows about it. Further, the amount of downlink signaling overhead is small, because only two parameters $\beta^\star$ and $x^\star$ are broadcasted to the sources. Finally, when the node density is high, the scalability of the network is a crucial concern and reporting $w_l$ and $b_l$ for each source is impractical. In this case, the AP can compute $\beta^\star$ and $x^\star$ by estimating the distribution of $w_l$ and $b_l$, as well as the number of source nodes, which reduces the uplink signaling overhead.

\subsubsection{Asymptotic Behavior of The Optimal Solution}
In the energy-adequate regime, the sleeping period parameter $r_l^\star \to \infty$ of source~$l$ as $t_s/\mathbb{E}[T] \to 0$, while the ratio $r^{\star}_l/r^{\star}_i$ between source~$l$ and source~$i$ is kept as a constant for all $l$ and $i$. In this case, the sleeping time of the sources tends to zero. Meanwhile, since $t_s/ \mathbb{E}[T] \to 0$, the sensing time becomes negligible. The channel access probability of source~$l$ in this limit can be computed as 
\begin{align}\label{sol_sum_bi_at_limit}
\lim_{\frac{t_s}{\mathbb{E}[T]}\to 0}\sigma^{\star}_l=\min\{b_l,\beta^{\star}\sqrt{w_l}\}. 
\end{align} 
Because of \eqref{condition_beta}, $\lim_{t_s/\mathbb{E}[T]\to 0} \sum_{i=1}^M\sigma_i^\star = 1$. Hence, the channel is occupied by the sources at all time, without any time overhead on sensing and sleeping. The performance of such scheduler is asymptotically no worse than any synchronized scheduler (e.g., for time-slotted systems) in theory, for which we assume that the time overhead needed for coordinating different sources with random packet sizes are omitted. Note that because of the coordination overhead, such synchronized schedulers are only feasible when the number of sources M is small.

 In synchronized schedulers, the AP assigns channel access among the sources in an i.i.d. manner. Under such a scheduler, there is a probability vector $\mathbf{a} = \left\{ a_l \right\}_{l=1}^M, \sum_{i=1}^M a_i = 1$, such that each source~$l$ gains channel access after a packet transmission with a probability equal to $a_l$. We can perform an analysis similar to that of Section~\ref{objective_function1}, and show that the total weighted average peak age of a synchronized scheduler is given by
\begin{equation}\label{obj_func_c}
\sum_{i=1}^M\left[\frac{w_i\mathbb{E}[T]}{a_i}+w_i~\mathbb{E}[T]\right].
\end{equation}
Moreover, similar to the derivation in \ifreport Appendix \ref{Appendix_A''}\fi, we can show that the fraction of time during which source~$l$ transmits update packets under a synchronized scheduler is equal to $a_l$. Hence, the problem of designing an optimal synchronized scheduler that minimizes the total weighted average peak age under energy constraints can be cast as the following convex optimization problem:
\begin{align}
\bar{\Delta}_{\text{opt-s}}^{\text{peak}}\triangleq\min_{a_i>0}&\sum_{i=1}^M\left[\frac{w_i}{a_i}+w_i\right]\label{problem_centr}\\
\textbf{s.t.}~&a_l\le b_l,~\forall l,\label{eq105}\\&
\sum_{i=1}^Ma_i=1,\label{eq106}
\end{align}
where we note that we have normalized the objective function by $\mathbb{E}[T]$. Next, we show that the performance of our proposed algorithm converges to that of the optimal synchronized scheduler when $t_s/\mathbb{E}[T]\to 0$. 
\begin{corollary}\label{corollary_centeralized_scheduler}
If $\sum_{i=1}^Mb_i\ge 1$, then we have
\begin{align}
\lim_{\frac{t_s}{\mathbb{E}[T]}\rightarrow 0} \bar{\Delta}^{\text{peak}}_{\text{opt}}=\bar{\Delta}_{\text{opt-s}}^{\text{peak}}.
\end{align}
\end{corollary}
\ifreport
\begin{proof}
The proof is provided in Appendix \ref{Appendix_E} which is listed at the end of the appendix as it requires some results from precedent appendixes.
\end{proof}
\fi
Synchronized schedulers were recently studied in \cite{talak2018optimizing2} for the case without energy constraints, i.e., $b_l \geq 1$ for all $l$. According to Corollary \ref{corollary_centeralized_scheduler}, the channel access probability of the synchronized scheduler in \cite{talak2018optimizing2} is a special case of our solution \eqref{sol_sum_bi_at_limit} where $b_l \geq 1$ for all $l$.

%
%

On the other hand, in the energy-scarce regime, the sleeping period parameter $r^{\star}_l$ of source~$l$ converges to a constant value when $t_s/\mathbb{E}[T]\to 0$, i.e., we have
\begin{align}
 \lim_{\frac{t_s}{\mathbb{E}[T]}\to 0}r^{\star}_l=\frac{b_l}{1-\sum_{i=1}^Mb_i}.
\end{align}
 Since the cumulative energy is scarce, the sources necessarily need to idle in order to meet their target lifetime. Hence, sleep periods are imposed for achieving the optimal trade-off between minimizing AoI and energy consumption.

\section{Proofs of the Main Results}\label{proof1}
\ifreport
In this section, we provide the proofs of Theorem \ref{thm_case_sum_bi_ge_1}, Corollary \ref{cor1}, Theorem \ref{thm_case_sum_bi_le_1}, and Corollary \ref{cor2}.
\subsection{The Proofs of Theorem \ref{thm_case_sum_bi_ge_1} and Corollary \ref{cor1}}\label{proof_case_sum_bi_ge_1}
We prove Theorem \ref{thm_case_sum_bi_ge_1} and Corollary \ref{cor1} in three steps: 

\textbf{Step 1}: We begin by showing that our solution $\mathbf{r}^{\star}$ given by \eqref{r_f_b_gr_1} - \eqref{condition_beta} is feasible for Problem \textbf{1}. 
\begin{lemma}\label{lemma_feasibility}
If $\sum_{i=1}^Mb_i\ge 1$, then the solution $\mathbf{r}^{\star}$ given by  \eqref{r_f_b_gr_1} - \eqref{condition_beta} is feasible for Problem \textbf{1}. 
\end{lemma}
 \ifreport
\begin{proof}
See Appendix \ref{Appendix_A}.
\end{proof}
\else 
\begin{proof}
See our technical report \cite{bedewy2019optimal}.
\end{proof}
\fi
Hence, by substituting this solution $\mathbf{r}^{\star}$ into the objective function of Problem \textbf{1} in \eqref{problem1}, we get an upper bound on the optimal value $\bar{\Delta}^{\text{peak}}_{\text{opt}}$, which is expressed in the following lemma: 
\begin{lemma}\label{lemma_upper_bound_sum_i_b_i_ge1}
If $\sum_{i=1}^Mb_i\ge 1$, then
\begin{align}\label{upper_bound_sumb_i_ge_1}
\bar{\Delta}^{\text{peak}}_{\text{opt}}\le \bar{\Delta}^{\text{peak}}(\mathbf{r}^{\star})\le \sum_{i=1}^M\left[\frac{w_ie^{x^{\star}\frac{t_s}{\mathbb{E}[T]}}\left(1+\frac{1}{x^{\star}}\right)}{\min\{b_i,\beta^{\star}\sqrt{w_i}\}}+w_i\right],
\end{align}
where $x^{\star}$, $\beta^{\star}$ are defined in \eqref{x*}, \eqref{condition_beta}. 
\end{lemma}
\begin{proof}
In Lemma \ref{lemma_feasibility}, we showed that our proposed solution $\mathbf{r}^{\star}$ given by \eqref{r_f_b_gr_1} - \eqref{condition_beta} is feasible for Problem \textbf{1}. Hence, we substitute this solution into Problem \textbf{1} to obtain the following upper bound:
\begin{align}
\sum_{i=1}^M\left[\frac{w_ie^{x^{\star}\frac{t_s}{\mathbb{E}[T]}}\left(1+\frac{1}{x^{\star}}\right)e^{-\min\{b_i,\beta^{\star}\sqrt{w_i}\}x^\star\frac{t_s}{\mathbb{E}[T]}}}{\min\{b_i,\beta^{\star}\sqrt{w_i}\}}+w_i\right].
\end{align}
Next, we replace $e^{-\min\{b_i,\beta^{\star}\sqrt{w_i}\}x^\star(t_s/\mathbb{E}[T])}$ by 1 
to derive another upper bound with a simple expression, which is given by \eqref{upper_bound_sumb_i_ge_1}. This completes the proof. 
\end{proof}

\textbf{Step 2}: 
We now construct a lower bound on the optimal
value of Problem \textbf{1}. Suppose that $\mathbf{r}=(r_1,\ldots,r_M)$ is a feasible solution to Problem \textbf{1}, such that $r_l >0$ and
\begin{align}
\frac{[1-e^{-r_l\frac{t_s}{\mathbb{E}[T]}}]\sum_{i=1}^Mr_i+r_le^{-r_l\frac{t_s}{\mathbb{E}[T]}}}{\sum_{i=1}^Mr_i+1}\leq b_l, \forall l.
\end{align} 
Because $[1-e^{-r_l(t_s/\mathbb{E}[T])}]\sum_{i=1}^Mr_i+r_le^{-r_l(t_s/\mathbb{E}[T])} > r_l$ for all $l$, $\mathbf{r}$ satisfies $r_l/(\sum_{i=1}^Mr_i+1)\leq b_l$. Hence, the following Problem \textbf{2} has a larger feasible set than Problem \textbf{1}: (Problem \textbf{2})
\begin{align}
\begin{split}\label{problem2}
\!\!\!\!\bar{\Delta}^{\text{peak}}_{\text{opt},2}\triangleq\min_{r_l>0}& \sum_{l=1}^M\frac{w_le^{-r_l\frac{t_s}{\mathbb{E}[T]}}}{r_l}e^{\sum_{i=1}^Mr_i\frac{t_s}{\mathbb{E}[T]}}\left(1+\sum_{i=1}^Mr_i\right)+\sum_{l=1}^Mw_l
\end{split}\\
\begin{split}\label{constrain_final_1}
\textbf{s.t.}~&r_l\leq b_l\left(\sum_{i=1}^Mr_i+1\right), \forall l,
\end{split}
\end{align}
where $\bar{\Delta}^{\text{peak}}_{\text{opt},2}$ is the optimal value of Problem \textbf{2}. The optimal objective value of Problem \textbf{2} is a lower bound of that of Problem \textbf{1}. We note that the constraint set corresponding to Problem~\textbf{2} is convex. Thus, this relaxation converts the constraint set of Problem~\textbf{1} to a convex one, and hence enables us to obtain a lower bound for the optimal value of Problem \textbf{1}, which is expressed in the following lemma: 
\begin{lemma}\label{lemma_bounds}
If $\sum_{i=1}^Mb_i\ge 1$, then  
\begin{align}
\bar{\Delta}^{\text{peak}}_{\text{opt}}\ge\bar{\Delta}^{\text{peak}}_{\text{opt},2}\ge\sum_{i=1}^M\left[\frac{w_i}{\min\{b_i,\beta^{\star}\sqrt{w_i}\}}+w_i\right],\label{lower_bound_sum_bi_ge_1}
\end{align}
where $\beta^{\star}$ is the root of \eqref{condition_beta}. 
\end{lemma}
\ifreport
\begin{proof}
See Appendix \ref{Appendix_B}.
\end{proof}
\else 
\begin{proof}
See our technical report \cite{bedewy2019optimal}.
\end{proof}
\fi

\textbf{Step 3}: After the upper and lower bounds of $\bar{\Delta}^{\text{peak}}_{\text{opt}}$ were derived in Steps 1-2, we are ready to analysis their gap. By combining \eqref{upper_bound_sumb_i_ge_1} and \eqref{lower_bound_sum_bi_ge_1}, the sub-optimality gap of the solution $\mathbf{r}^\star$ given by \eqref{r_f_b_gr_1} - \eqref{condition_beta} is upper bounded by
\begin{align}\label{G1}
\begin{split}
&\vert\bar{\Delta}^{\text{peak}}(\mathbf{r}^{\star})\!-\!\bar{\Delta}^{\text{peak}}_{\text{opt}}\vert\le \sum_{i=1}^M\frac{w_i\left(e^{x^{\star}\frac{t_s}{\mathbb{E}[T]}}(1\!+\!\frac{1}{x^{\star}})\!-\!1 \right)}{\min\{b_i,\beta^{\star}\sqrt{w_i}\}},
\end{split}
\end{align}
where $x^{\star}$, $\beta^{\star}$ are defined in \eqref{x*}, \eqref{condition_beta}. Next, we characterize the right-hand-side (RHS) of \eqref{G1} by Taylor expansion. For simplicity, let $\epsilon=\frac{t_s}{\mathbb{E}[T]}$. Using the expression for $x^{\star}$ from \eqref{x*}, we have
\begin{equation}\label{x^*t_s}
\begin{split}
x^{\star}\epsilon=&-\frac{\epsilon}{2}+\sqrt{\frac{\epsilon^2}{4}+\epsilon}
=\frac{\epsilon}{\frac{\epsilon}{2}+\sqrt{\frac{\epsilon^2}{4}+\epsilon}}
=\sqrt{\epsilon}+o(\sqrt{\epsilon}).
\end{split}
\end{equation}
Moreover,
\begin{equation}\label{x^*ET}
\begin{split}
x^\star=&-\frac{1}{2}+\sqrt{\frac{1}{4}+\frac{1}{\epsilon}}
=\frac{\frac{1}{\epsilon}}{\frac{1}{2} + \sqrt{\frac{1}{4} + \frac{1}{\epsilon}}}
=\frac{1}{\sqrt{\epsilon}} + o\left(\frac{1}{\sqrt{\epsilon}}\right). 
\end{split}
\end{equation}
 Substituting \eqref{x^*t_s} and~\eqref{x^*ET} in \eqref{G1}, we obtain
\begin{align}
\vert\bar{\Delta}^{\text{peak}}(\mathbf{r}^{\star})-\bar{\Delta}^{\text{peak}}_{\text{opt}}\vert&
\le \sum_{i=1}^M\frac{w_i[e^{\sqrt{\epsilon} + o(\sqrt{\epsilon})} (1 + \sqrt{\epsilon} + o(\sqrt{\epsilon})) - 1 ]}{\min\{b_i,\beta^{\star}\sqrt{w_i}\}}\nonumber\\&
=\sum_{i=1}^M\frac{w_i[(1\!+\! \sqrt{\epsilon}\! +\!o(\sqrt{\epsilon}))(1 \!+\! \sqrt{\epsilon} \!+\! o(\sqrt{\epsilon})) \!-\! 1 ]}{\min\{b_i,\beta^{\star}\sqrt{w_i}\}}\nonumber\\&
=2 \sqrt{\epsilon} \sum_{i=1}^M\frac{w_i}{\min\{b_i,\beta^{\star}\sqrt{w_i}\}}  + o(\sqrt{\epsilon}),
\end{align}
where the second inequality involves the use of Taylor expansion. This proves Theorem \ref{thm_case_sum_bi_ge_1}. 

Moreover, we can observe that the gap $\vert\bar{\Delta}^{\text{peak}}(\mathbf{r}^{\star})-\bar{\Delta}^{\text{peak}}_{\text{opt}}\vert$ in the energy-adequate regime converges to zero at a speed of $O(\sqrt{\epsilon})$, as $\epsilon \to 0$. We also observe that both the upper and lower bounds \eqref{upper_bound_sumb_i_ge_1}, \eqref{lower_bound_sum_bi_ge_1}, converge to $\sum_{i=1}^M[(w_i/\min\{b_i,\beta^{\star}\sqrt{w_i}\})+w_i]$ as $t_s/\mathbb{E}[T]\rightarrow 0$. Thus, this value is the asymptotic optimal value of Problem \textbf{1}. This proves Corollary \ref{cor1}. 

\fi

\ifreport
\subsection{The Proofs of Theorem \ref{thm_case_sum_bi_le_1} and Corollary \ref{cor2}}\label{proof_case_sum_bi_le_1}
Similar to Section \ref{proof_case_sum_bi_ge_1}, we prove Theorem \ref{thm_case_sum_bi_le_1} and
Corollary \ref{cor2} in also three steps:

\textbf{Step 1}: We show that the proposed solution $\mathbf{r}^{\star}$ given by \eqref{r_f_b_gr_1} and  \eqref{condition_x*_beta*} - \eqref{feas_fact_eq27} is a feasible solution for Problem \textbf{1}.
\begin{lemma}\label{lemma_feasibility2}
If $\sum_{i=1}^Mb_i< 1$, then the solution $\mathbf{r}^{\star}$ given by \eqref{r_f_b_gr_1} and  \eqref{condition_x*_beta*} - \eqref{feas_fact_eq27} is feasible for Problem \textbf{1}.
\end{lemma}
\begin{proof}
See Appendix \ref{Appendix_C}.
\end{proof}
Now, we construct an upper bound on the optimal value of Problem \textbf{1} using our proposed solution as follows:
\begin{lemma}\label{lemma_upper_bound_sum_i_b_i_le1}
If $\sum_{i=1}^Mb_i< 1$, then 
\begin{align}\label{upper_bound_sumb_i_le_1}
\begin{split}
\bar{\Delta}^{\text{peak}}_{\text{opt}}\le\bar{\Delta}^{\text{peak}}(\mathbf{r}^{\star})\le& \sum_{l=1}^M\!\frac{w_l}{b_l}e^{\sum_{i=1}^Mb_ix^{\star}\frac{t_s}{\mathbb{E}[T]}}\left(\frac{1}{x^{\star}}+\sum_{i=1}^Mb_i\!\right)\\&+\sum_{l=1}^Mw_l,
\end{split}
\end{align}
where $x^{\star}$ is defined in \eqref{condition_x*_beta*}. 
\end{lemma}
\begin{proof}
In Lemma \ref{lemma_feasibility2}, we showed that our proposed solution $\mathbf{r}^{\star}$ given by  \eqref{r_f_b_gr_1} and  \eqref{condition_x*_beta*} - \eqref{feas_fact_eq27} is feasible for Problem \textbf{1}. Hence, we substitute this solution into Problem \textbf{1} to obtain the following upper bound:
\begin{align}
\sum_{l=1}^M\!\frac{w_le^{-b_lx^\star\frac{t_s}{\mathbb{E}[T]}}}{b_l}e^{\sum_{i=1}^Mb_ix^{\star}\frac{t_s}{\mathbb{E}[T]}}\left(\frac{1}{x^{\star}}+\sum_{i=1}^Mb_i\!\right)+\sum_{l=1}^Mw_l.
\end{align}
Next, we replace $e^{-b_lx^\star\frac{t_s}{\mathbb{E}[T]}}$ by 1 
to derive another upper bound with a simple expression, which is given by \eqref{upper_bound_sumb_i_le_1}. This completes the proof. 
\end{proof}

\textbf{Step 2}: Similar to the proof in Section \ref{proof_case_sum_bi_ge_1}, we use the relaxed problem, Problem \textbf{2}, to construct a lower bound as follows:
\begin{lemma}\label{lemma_lower_bound_sum_i_b_i_le_1}
If $\sum_{i=1}^Mb_i< 1$, then 
\begin{align}
\bar{\Delta}^{\text{peak}}_{\text{opt}}\ge\bar{\Delta}^{\text{peak}}_{\text{opt},2}\ge\sum_{l=1}^M\frac{w_l}{b_l}e^{\frac{-\sum_{i=1}^Mb_i}{1-\sum_{i=1}^Mb_i}\frac{t_s}{\mathbb{E}[T]}}+\sum_{l=1}^Mw_l.\label{lower_bound_sum_bi_le_1}
\end{align}
\end{lemma}
\begin{proof}
See Appendix \ref{Appendix_D}.
\end{proof}

\textbf{Step 3}: We now characterize the sub-optimality gap by analyzing the upper and lower bounds constructed above. By combining \eqref{upper_bound_sumb_i_le_1} and \eqref{lower_bound_sum_bi_le_1}, the sub-optimality gap of the solution $\mathbf{r}^\star$ given by \eqref{r_f_b_gr_1} and  \eqref{condition_x*_beta*} - \eqref{feas_fact_eq27} is upper bounded by
\begin{align}\label{G2}
\begin{split}
&\vert\bar{\Delta}^{\text{peak}}(\mathbf{r}^{\star})-\bar{\Delta}^{\text{peak}}_{\text{opt}}\vert\\&\le \sum_{l=1}^M\frac{w_l}{b_l}\!\left[ e^{\sum_{i=1}^Mb_ix^{\star}\frac{t_s}{\mathbb{E}[T]}}\left(\frac{1}{x^{\star}}\!+\!\sum_{i=1}^Mb_i\right)\!-\!e^{\frac{-\sum_{i=1}^Mb_i}{1-\sum_{i=1}^Mb_i}\frac{t_s}{\mathbb{E}[T]}}\right].
\end{split}
\end{align}
where $x^{\star}$ is defined in \eqref{condition_x*_beta*}. Next, we characterize the RHS of \eqref{G2} by Taylor expansion. For simplicity, let $\epsilon=t_s/\mathbb{E}[T]$, $Z=(\sum_{i=1}^Mb_i)/(1-\sum_{i=1}^Mb_i)$, and  $k_l=(\sum_{i=1}^Mb_i-b_l)/(1-\sum_{i=1}^Mb_i)^2$. Using Taylor expansion, we are able to obtain the following:
\begin{align}
&\min_l c_l=1+\left(\min_l k_l\right)\epsilon+o(\epsilon),\label{eq37}\\&
\frac{1}{\min_lc_l}=\max_l\frac{1}{c_l}=1+\left(\max_lk_l\right)\epsilon+o(\epsilon).\label{eq38}
\end{align}
Using \eqref{eq37}, \eqref{eq38}, $x^{\star}$ from \eqref{condition_x*_beta*}, and Taylor expansion again, we get
\begin{align}
\begin{split}
e^{\sum_{i=1}^Mb_ix^{\star}\epsilon}&=1+Z\left(1+\left(\min_lk_l\right)\epsilon+o(\epsilon)\right)\epsilon+o(\epsilon)\\&=1+Z\epsilon+o(\epsilon),
\end{split}\label{eq39}\\
\begin{split}
\frac{1}{x^{\star}}+\sum_{i=1}^Mb_i&=\frac{1-\sum_{i=1}^Mb_i}{\min_lc_l}+\sum_{i=1}^Mb_i\\&=1+\left(\max_lk_l\right)\left(1-\sum_{i=1}^Mb_i\right)\epsilon+o(\epsilon),
\end{split}\\
e^{-Z\epsilon}&=1-Z\epsilon+o(\epsilon).\label{eq41}
\end{align}
Substituting \eqref{eq39} - \eqref{eq41} into \eqref{G2}, we get \eqref{sub_opt_gap_eq_2}. This proves Theorem \eqref{thm_case_sum_bi_le_1}. 

Moreover, we observe that the gap $\vert\bar{\Delta}^{\text{peak}}(\mathbf{r}^{\star})-\bar{\Delta}^{\text{peak}}_{\text{opt}}\vert$ in the energy-scarce regime converges to zero at a speed of $O(\epsilon)$, as $\epsilon \to 0$. We also observe that both the upper and lower bounds \eqref{upper_bound_sumb_i_le_1}, \eqref{lower_bound_sum_bi_le_1}, converge to $\sum_{i=1}^M[(w_i/b_i)+w_i]$ as $t_s/\mathbb{E}[T]\to 0$. Thus, this value is the asymptotic optimal value of Problem \textbf{1} in this case. This proves Corollary \ref{cor2}.
\fi
\section{Numerical Results}\label{Simulations}
We use Matlab to evaluate the performance of our algorithm.
We use ``Age-optimal scheduler'' to denote the sleep-wake scheduler with the sleep period paramters $r_l^{\star}$'s as in \eqref{r_f_b_gr_1}, which was shown to be near-optimal in Theorem \ref{thm_case_sum_bi_ge_1} and Theorem \ref{thm_case_sum_bi_le_1}. By ``Throughput-optimal scheduler'', we  refer to the sleep-wake algorithm of~\cite{chen2013life} that is known to achieve the optimal trade-off between the throughput and energy consumption reduction. Moreover, we use ``Fixed sleep-rate scheduler'' to denote the sleep-wake scheduler in which the sleep period parameters $r_l$'s are equal for all the sources, i.e., $r_l=k$ for all $l$,  where the parameter $k$ has been chosen so as to satisfy the energy constraints of Problem \textbf{1}.
We also let $\bar{\Delta}_{\text{un}}^{\text{peak}}(\mathbf{r})$ denote the unnormalized total weighted average peak age in \eqref{t_avg_peak_age}. Finally, we would like to mention that we do not compare the performance of our proposed algorithm with the CSMA algorithms of~\cite{maatouk2019minimizing,wang2019broadcast} since the objective of these works was solely to minimize the age. Since they do not incorporate energy constraints, it is not fair to compare the performance of our algorithm with them. 

Unless stated otherwise, our set up is as follows: The average transmission time is $\mathbb{E}[T]=5$ ms. The weights $w_l$'s attached to different sources are generated by sampling from a uniform distribution in the interval $[0, 10]$. The target energy efficiencies $b_l$'s are randomly generated uniformly  within the range $[0, 1]$. 


\begin{figure}[t]
\centering
\includegraphics[scale=0.296]{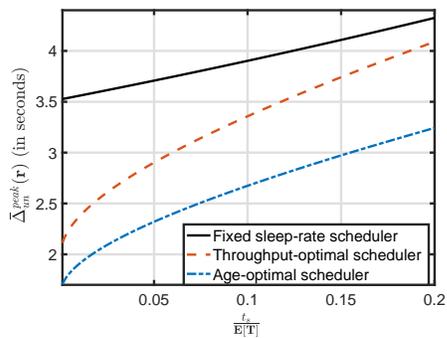}
\centering
\captionsetup{justification=justified}
\caption{Total weighted average peak age $\bar{\Delta}_{\text{un}}^{\text{peak}}(\mathbf{r})$ in \eqref{t_avg_peak_age} versus the ratio $\frac{t_s}{\mathbb{E}[T]}$ for $M=10$ sources.}
\label{peak_age_comp}
\end{figure}
We set the number of sources at $M=10$. 
Figure~\ref{peak_age_comp} plots the total weighted average peak age $\bar{\Delta}_{\text{un}}^{\text{peak}}(\mathbf{r})$ in  \eqref{t_avg_peak_age} as a function of the ratio $\frac{t_s}{\mathbb{E}[T]}$. The age-optimal scheduler is seen to outperform the throughput-optimal and Fixed sleep-rate schedulers. This implies that what minimizes the throughput does not necessarily minimize AoI and vice versa. 
 Moreover, we observe that the total weighted average peak age of all schedulers increases as the sensing time increases. This is expected since an increase in the sensing time leads to an increase in the probability of packet collisions, which in turn deteriorates the age performance of these schedulers.

\begin{figure}[t]
\centering
\includegraphics[scale=0.296]{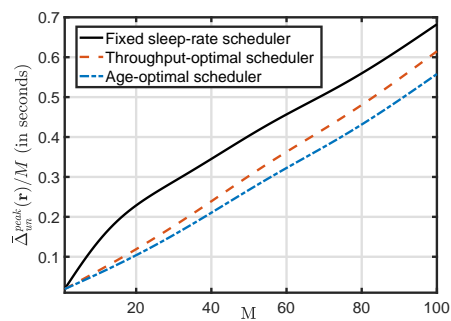}
\centering
\captionsetup{justification=justified}
\caption{Total weighted average peak age $\bar{\Delta}_{\text{un}}^{\text{peak}}(\mathbf{r})$ in \eqref{t_avg_peak_age} versus the number of sources $M$, where $\bar{\Delta}_{\text{un}}^{\text{peak}}(\mathbf{r})$ has been normalized by $M$ while plotting.}
\label{peak_age_comp_num}
\end{figure}
We then scale the number of sources $M$, and plot $\bar{\Delta}_{\text{un}}^{\text{peak}}(\mathbf{r})$  in  \eqref{t_avg_peak_age} as a function of $M$ in Figure~\ref{peak_age_comp_num}. While plotting, we normalize the performance by the number of sources $M$. The sensing time $t_s$ is fixed at $t_s=40~\mu$s. The weights $w_l$'s corresponding to different sources are randomly generated uniformly  within the range $[0, 2]$. The age-optimal scheduler is shown to outperform other schedulers uniformly for all values of $M$. Moreover, as we can observe, the average peak age of the sources under age-optimal scheduler increases up to around 0.55 seconds only, while the number of sources rises from 1 to 100. This indicates the  robustness of our algorithm to changes in the number of sources in a network.

\begin{figure}[t]
\centering
\includegraphics[scale=0.296]{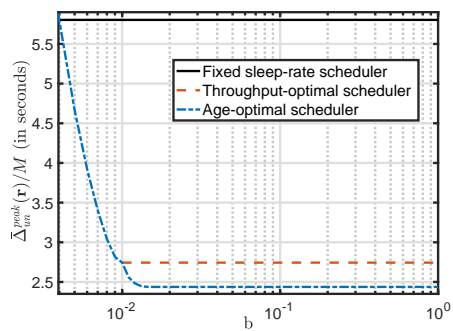}
\centering
\captionsetup{justification=justified}
\caption{Total weighted average peak age $\bar{\Delta}_{\text{un}}^{\text{peak}}(\mathbf{r})$ in \eqref{t_avg_peak_age} versus the target energy efficiency $b$ for  $M=100$ sources, where $\bar{\Delta}_{\text{un}}^{\text{peak}}(\mathbf{r})$ has been normalized by $M$ while plotting.}
\label{peak_age_comp_lower_bound}
\end{figure}
In Figure \ref{peak_age_comp_lower_bound}, we fix the value of $M$ at $100$ sources and the target energy efficiencies at the same value for all the sources, i.e., $b_l= b$ for all $l$. We then vary the parameter $b$ and plot the resulting performances. While plotting, we normalize the performance by the number of sources $M$.  We exclude the simulation of the throughput-optimal scheduler for $b< 0.01$ since the sleeping period parameters that are proposed in \cite{chen2013life} are not feasible for Problem \textbf{1} in energy-scarce regime, i.e., when $\sum_{i=1}^Mb_i<1$. The age-optimal scheduler outperforms the rest of the schedulers. Moreover, its performance is a decreasing function of $b$, and then settles at a constant value. This occurs because we observe from~\eqref{r_f_b_gr_1} that there exists a value for $b$ after which our proposed solution value,  $\mathbf{r}^{\star}$, is a function solely of weights $w_l$'s and $\beta^\star$, and not of $b$. Thus, the performance of the proposed scheduler saturates after this value of $b$.

\begin{figure}[t]
\centering
\includegraphics[scale=0.296]{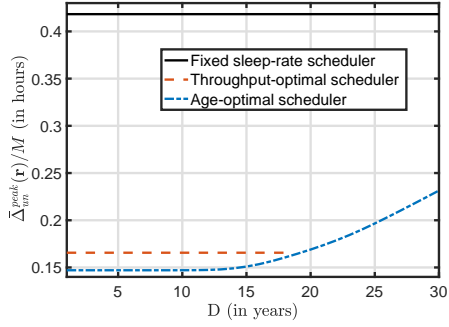}
\centering
\captionsetup{justification=justified}
\caption{Total weighted average peak age $\bar{\Delta}_{\text{un}}^{\text{peak}}(\mathbf{r})$ in \eqref{t_avg_peak_age} versus the target lifetime $D$ for a dense network with number of sources $M=10^5$, where $\bar{\Delta}_{\text{un}}^{\text{peak}}(\mathbf{r})$ has been normalized by $M$ while plotting. Since the throughput--optimal scheduler is infeasible for values of $D$ greater than $18$ years, we do not plot its performance for these values.}
\label{peak_age_comp_dense_network}
\end{figure}
We now show the effectiveness of the proposed scheduler when deployed in ``dense networks''~\cite{kowshik2019energy,kowshik2019fundamental}. Dense networks are characterized by a large number of sources connected to a single AP. We fix $M$ at $10^5$ sources, and take the target lifetimes of the sources to be equal, i.e., $D_l= D$ for all $l$. The weights $w_l$'s corresponding to different sources are generated randomly by sampling from the uniform distribution in the range $[0, 2]$. We let the initial battery level $B_l= 8$ mAh for all $l$ and the output voltage is 5 Volt. We also let the energy consumption in a transmission mode to be 24.75 mW for all sources. We vary the parameter $D$ and plot the resulting performance in Figure \ref{peak_age_comp_dense_network}. While plotting, we normalize the performance by the number of sources $M$. We exclude simulations for the throughput-optimal scheduler for values of $D$ for which the scheduler is infeasible, i.e., its cumulative energy consumption exceeds the total allowable energy consumption. The age-optimal scheduler is seen to outperform the others. 
As observed in Figure~\ref{peak_age_comp_dense_network}, under the age-optimal scheduler, sources can be active for up to 25 years, while simultaneously achieving a decent average peak age of around .2 hour, i.e., 12 minutes. This makes it apt for dense networks, where it is crucial that the sources are necessarily active for many years.
\subsection{NS-3 Simulation}\label{ns3sim}
We use NS-3 \cite{ns_3_2_0} to investigate the effect of our assumptions on the performance of the age-optimal scheduler in a more practical situation. We simulate the Age-optimal scheduler by using IEEE 802.11b by disabling the RTS-CTS and modifying the back-off times to be exponentially distributed in the MAC layer. Our simulation results are averaged over 5 system realizations. The UDP saturation conditions are satisfied such that all source nodes always have a packet to send.

 Our simulation consists of a WiFi network with 1 AP and 3 associated source nodes in a field of size 50m $\times$ 50m. We set the sensing threshold to -100 dBm which covers a range of 110m. Thus, all sources can hear each other. We set the initial battery level of all source to be 60 mAh, where the output voltage is 5 Volt. For each source, the power consumption in the transmission mode is 24.75 mW, and the power consumption in the sleep mode is 15 $\mu$W. Moreover, all weights are set to unity, i.e., $w_l=1$ for all $l$.

\begin{figure}[t]
\centering
\includegraphics[scale=0.296]{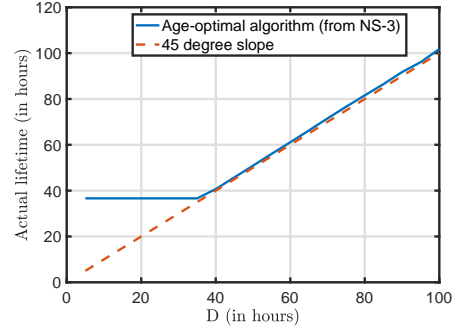}
\centering
\captionsetup{justification=justified}
\caption{The average actual lifetime versus the target lifetime $D$.}
\label{ns3_actual_lifetime}
\end{figure}
Figure \ref{ns3_actual_lifetime} plots the average actual lifetime of the sources versus the target lifetime, where we take the target lifetimes of all sources to be equal, i.e., $D_l=D$ for all $l$. As we can observe, the actual lifetime of the age-optimal scheduler always achieves the target lifetime. This suggests that our assumptions (i.e.,  (i) omitting the power dissipation in the sleep mode and in the sensing times, (ii) the average transmission times and collision times  are equal to each other) do not affect the performance of the algorithm which reaches its target lifetime.

\begin{figure}[t]
\centering
\includegraphics[scale=0.296]{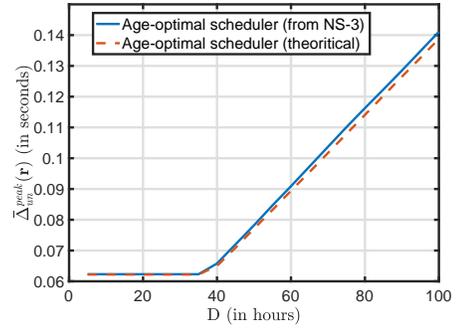}
\centering
\captionsetup{,justification=justified}
\caption{Total weighted average peak age $\bar{\Delta}_{\text{un}}^{\text{peak}}(\mathbf{r})$ versus the target lifetime $D$.}
\label{ns3_total_avg_peak}
\end{figure}
Figure \ref{ns3_total_avg_peak} plots the total weighted average peak age versus the target lifetime, where again we take the target lifetimes of all sources to be equal, i.e., $D_l=D$ for all $l$. The age-optimal scheduler (theoretical) curve is obtained using  \eqref{t_avg_peak_age}, while
the age-optimal scheduler (from NS-3) curve is obtained using the NS-3 simulator.  As we can observe, the difference between the plotted curves does not exceed 2\% of the age-optimal scheduler (theoretical) performance.
This emphasizes the negligible impact of our assumptions on the performance of our proposed algorithm.

\section{Conclusions}\label{Concl}
We designed an efficient sleep-wake mechanism for wireless networks that attains the optimal trade-off between minimizing the AoI and energy consumption. Since the associated optimization problem is non-convex, in general we could not hope to solve it for all values of the system parameters. However, in the regime when the carrier sensing time $t_s$ is negligible as compared to the average transmission time $\mathbb{E}[T]$, we were able to provide a near-optimal solution. Moreover, the proposed solution is on a simple form that allowed us to design a simple-to-implement algorithm to obtain its value. Finally, we showed that, in the energy-adequate regime, the performance of our proposed algorithm is asymptotically no worse than that of the optimal synchronized scheduler, as $t_s/\mathbb{E}[T]\to 0$. 


\begin{acks}
The authors appreciate Jiayu Pan and Shaoyi Li for their great efforts in obtaining the ns-3 simulation results. 
\end{acks} 
\bibliographystyle{ACM-Reference-Format}
\bibliography{MyLib}
\ifreport 
\balance
\section{Appendix}
\appendix
\section{Derivation of (\ref{access_prob_in_agiven_cycle})}\label{Appendix_A'}
Define $S_l$ as the residual sleeping period of source~$l$ after a sleep-wake cycle is over. Due to the memoryless property of exponential distribution, since the sleeping period of source~$l$ is exponentially distributed with mean value $\mathbb{E}[T]/r_l$, $S_l$ is also exponentially distributed with mean value $\mathbb{E}[T]/r_l$. According to the proposed sleep-wake scheduler, source~$l$ gains access to the channel and transmits successfully in a given cycle if $S_i\ge S_l+t_s$ for all $i\neq l$. Hence, we have
\begin{align}
\alpha_l&=\mathbb{P}(S_i\ge S_l+t_s,~\forall i\neq l)\\&
\stackrel{(a)}{=}\mathbb{E}[\mathbb{P}(S_i\ge S_l+t_s,~\forall i\neq l\vert S_l)]\\&
\stackrel{(b)}{=}\mathbb{E}\left[\prod_{i\neq l}\mathbb{P}(S_i\ge S_l+t_s\vert S_l)\right]\\&
=\int_0^\infty\left[\prod_{i\neq l}e^{-r_i\frac{s_l+t_s}{\mathbb{E}[T]}}\right]\frac{r_l}{\mathbb{E}[T]}e^{-r_l\frac{s_l}{\mathbb{E}[T]}}ds_l\\&
=\frac{r_l e^{r_l\frac{t_s}{\mathbb{E}[T]}}}{e^{\sum_{i=1}^Mr_i\frac{t_s}{\mathbb{E}[T]}}\sum_{i=1}^Mr_i},
\end{align}
where (a) is due to $\mathbb{P}[A]=\mathbb{E}[\mathbb{P}(A\vert B)]$, and (b) is due to the fact that $S_l$ is independent for different sources.

\section{Derivation of (\ref{sigma_l})}\label{Appendix_A''}
Recall the definition of $S_l$ at the beginning of Appendix \ref{Appendix_A'}. Moreover, define $P_l$ as the probability that source~$l$ transmits a packet in a given cycle, regardless whether packet collision occurs or not. For the sleep-wake mechanism the we are utilizing here, source~$l$ transmits in a given cycle as long as no other source wakes up before $S_l-t_s$, i.e., $S_i\ge S_l-t_s$ for all $i\neq l$. Hence, we have
\begin{align}
P_l&=\mathbb{P}(S_i\ge S_l-t_s,~\forall i\neq l)\\&
=\mathbb{P}(S_i\ge S_l-t_s,~\forall i\neq l,~S_l\ge t_s)+\mathbb{P}(S_l<t_s),\label{P_l}
\end{align}
where the first term in the RHS is given by
\begin{align}
&\mathbb{P}(S_i\ge S_l-t_s\ge 0,~\forall i\neq l)\\
=&\mathbb{E}[\mathbb{P}(S_i\ge S_l-t_s\ge 0,~\forall i\neq l\vert S_l)]\\
=&\mathbb{E}\left[\prod_{i\neq l}\mathbb{P}(S_i\ge S_l-t_s\ge 0\vert S_l)]\right]\\
=&\int_{t_s}^\infty\left[\prod_{i\neq l}e^{-r_i\frac{s_l-t_s}{\mathbb{E}[T]}}\right]\frac{r_l}{\mathbb{E}[T]}e^{-r_l\frac{s_l}{\mathbb{E}[T]}}ds_l\\
=&e^{-r_l\frac{t_s}{\mathbb{E}[T]}}\frac{r_l}{\sum_{i=1}^Mr_i}.\label{App15_56}
\end{align}
Since $S_l$ is exponentially distributed with mean value $\mathbb{E}[T]/r_l$, we can determine the second term in the RHS of \eqref{P_l} as follows:
\begin{align}\label{eq69_2}
\mathbb{P}(S_l<t_s)=1-e^{-r_l\frac{t_s}{\mathbb{E}[T]}}.
\end{align}
Substituting \eqref{App15_56} and \eqref{eq69_2} back into \eqref{P_l}, we get
\begin{align}
P_l=1-e^{-r_l\frac{t_s}{\mathbb{E}[T]}}+e^{-r_l\frac{t_s}{\mathbb{E}[T]}}\frac{r_l}{\sum_{i=1}^Mr_i}.
\end{align}
Let $\alpha_{\text{col}}$ denote the collision probability in a given cycle. We have $\alpha_{\text{col}}=1-\sum_{i=1}^M\alpha_i$, because each cycle includes either a successful transmission or a collision. Moreover, let $\mathbb{E}[\mathbf{Idle}]$ denote the mean of the idle duration in a cycle. By the renewal theory in stochastic processes \cite{gallager2012discrete}, $\sigma_l$ is given by
\begin{align}
\sigma_l&=\frac{P_l\mathbb{E}[T]}{(\sum_{i=1}^M\alpha_i+\alpha_{\text{col}})\mathbb{E}[T]+\mathbb{E}[\mathbf{Idle}]}\\&
=\frac{P_l\mathbb{E}[T]}{\mathbb{E}[T]+\frac{\mathbb{E}[T]}{\sum_{i=1}^Mr_i}}\\&
=\frac{[1-e^{-r_l\frac{t_s}{\mathbb{E}[T]}}]\sum_{i=1}^Mr_i+r_le^{-r_l\frac{t_s}{\mathbb{E}[T]}}}{\sum_{i=1}^Mr_i+1}.
\end{align}

\section{Proof of Lemma \ref{lemma_feasibility}}\label{Appendix_A}
First of all, we need to show that \eqref{condition_beta} has a solution for $\beta^{\star}$.
\begin{lemma}\label{lemma_solution_19}
Suppose that $w_l>0$, and $b_l>0$ for all $l$. If $\sum_{i=1}^M b_i \geq 1$, then \eqref{condition_beta} has a unique solution on $[0,\max_l(b_l/\sqrt{w_l})]$; otherwise, \eqref{condition_beta} has no solution.  
\end{lemma}
\begin{proof}
It is clear that if $\sum_{i=1}^M b_i =1$, then $\beta^\star$ satisfies \eqref{condition_beta} if and only if $\beta^\star\ge \max_l(b_l/\sqrt{w_l})$. Hence, \eqref{condition_beta} has a unique solution on $[0,\max_l(b_l/\sqrt{w_l})]$ in this case. We now focus on the case of $\sum_{i=1}^M b_i >1$. In this case, we have the following: 
\begin{itemize}
\item If $\beta^\star = 0$, then $\sum_{i=1}^M\min\{b_i,\beta^{\star}\sqrt{w_i}\}=0$.
\item If $\beta^\star =\max_l(b_l/\sqrt{w_l})$, then $\sum_{i=1}^M\min\{b_i,\beta^{\star}\sqrt{w_i}\}>1$. 
\item The left hand side (LHS) of \eqref{condition_beta} is strictly increasing and continuous in $\beta^\star$ on $[0,\max_l(b_l/\sqrt{w_l})]$.
\end{itemize}
As a result, \eqref{condition_beta} has a unique solution on $[0,\max_l(b_l/\sqrt{w_l})]$ in this case as well. Finally, if $\sum_{i=1}^M b_i <1$, then $\sum_{i=1}^M\min\{b_i,\beta^{\star}\sqrt{w_i}\}\le \sum_{i=1}^M b_i < 1$. Hence, \eqref{condition_beta} has no solution if $\sum_{i=1}^M b_i <1$. This completes the proof. 
\end{proof}
Since we have $\sum_{i=1}^M b_i \geq 1$, Lemma \ref{lemma_solution_19} implies that \eqref{condition_beta} has a solution for $\beta^\star$. Now, we are ready to prove Lemma \ref{lemma_feasibility}. Consider the following constraints:
\begin{align}\label{equivelent_constrant}
\frac{r_l\frac{t_s}{\mathbb{E}[T]}\sum_{i=1}^Mr_i+r_l}{\sum_{i=1}^Mr_i+1}\leq b_l, \forall l.
\end{align}
Since we have 
\begin{align}
&1-e^{-r_l\frac{t_s}{\mathbb{E}[T]}}\le r_l\frac{t_s}{\mathbb{E}[T]},\\&
e^{-r_l\frac{t_s}{\mathbb{E}[T]}}\leq 1,
\end{align}
then, 
\begin{align}
[1-e^{-r_l\frac{t_s}{\mathbb{E}[T]}}]\sum_{i=1}^Mr_i+r_le^{-r_l\frac{t_s}{\mathbb{E}[T]}}\leq r_l\frac{t_s}{\mathbb{E}[T]}\sum_{i=1}^Mr_i+r_l.
\end{align}
Thus, if the constraints in \eqref{equivelent_constrant} are satisfied for a given solution $\mathbf{r}$, then the constraints of Problem \textbf{1} are satisfied as well. We can observe that the constraints in \eqref{equivelent_constrant} are equivalent to the following set of constraints:
\begin{align}\label{inequality_feasible}
\begin{split}
&r_l\le b_l\frac{x+1}{1+\frac{t_s}{\mathbb{E}[T]}x},\forall l\\ 
&\sum_{i=1}^Mr_i=x.
\end{split}
\end{align}
Now, it is easy to show that if $x\le \sqrt{\mathbb{E}[T]/t_s}$, then  $x\le (x+1)/[1+(t_s/\mathbb{E}[T])x]$. Meanwhile, our proposed solution $\mathbf{r}^{\star}$ given by \eqref{r_f_b_gr_1} - \eqref{condition_beta} satisfies $\sum_{i=1}^Mr_i^{\star}=x^{\star}$. Thus, if we can show that $x^{\star}\le \sqrt{\mathbb{E}[T]/t_s}$, then
\begin{align}
r_l^{\star}=\min\{b_l,\beta^{\star}\sqrt{w_l}\}x^{\star}\le b_lx^{\star}\le b_l \frac{x^{\star}+1}{1+\frac{t_s}{\mathbb{E}[T]}x^{\star}},
\end{align}
and the constraints in \eqref{inequality_feasible} hold for our proposed solution $\mathbf{r}^\star$. What remains is to prove that $x^{\star}\le \sqrt{\mathbb{E}[T]/t_s}$. We have
\begin{align}
x^{\star}=&\frac{-1}{2}+\sqrt{\frac{1}{4}+\frac{\mathbb{E}[T]}{t_s}}\\
=& \frac{\frac{\mathbb{E}[T]}{t_s}}{\frac{1}{2}+\sqrt{\frac{1}{4}+\frac{\mathbb{E}[T]}{t_s}}}\\
\le& \frac{\frac{\mathbb{E}[T]}{t_s}}{\sqrt{\frac{\mathbb{E}[T]}{t_s}}}= \sqrt{\frac{\mathbb{E}[T]}{t_s}}.
\end{align}
Hence, our proposed solution $\mathbf{r}^{\star}$ given by \eqref{r_f_b_gr_1} - \eqref{condition_beta} satisfies  \eqref{inequality_feasible}, which implies \eqref{equivelent_constrant}. This completes the proof.

\section{Proof of Lemma \ref{lemma_bounds}}\label{Appendix_B}
By replacing $e^{-r_l(t_s/\mathbb{E}[T])}e^{\sum_{i=1}^Mr_i(t_s/\mathbb{E}[T])}$ in \eqref{problem2} of Problem \textbf{2} by 1, we obtain the following optimization problem:
\begin{align}
\begin{split}\label{lower_bound_optimization_problme1}
\min_{r_l>0}& \sum_{l=1}^M\frac{w_l}{r_l}\left(1+\sum_{i=1}^Mr_i\right)+\sum_{l=1}^Mw_l
\end{split}\\
\textbf{s.t.}~&r_l\leq b_l\left(\sum_{i=1}^Mr_i+1\right), \forall l.
\end{align}
Since $e^{-r_l(t_s/\mathbb{E}[T])}e^{\sum_{i=1}^Mr_i(t_s/\mathbb{E}[T])}\geq 1$, Problem \eqref{lower_bound_optimization_problme1} serves as a lower bound of Problem \textbf{2}, and hence a lower bound of Problem \textbf{1} as well. Define an auxiliary variable $y=\sum_{i=1}^Mr_i+1$. By this, we solve a two-layer nested optimization problem. In the inner layer, we optimize $\mathbf{r}$ for a given $y$. After solving $\mathbf{r}$, we will optimize $y$ in the outer layer. Now, fix the value of $y$, we obtain the following optimization problem (the inner layer):
\begin{align}
\begin{split}\label{optimization_auxiliary2}
\min_{r_i>0}& \sum_{i=1}^M\left[\frac{w_iy}{r_i}+w_i\right]
\end{split}\\
\textbf{s.t.}~&r_l\leq b_ly, \forall l,\label{const12}\\&
\sum_{i=1}^Mr_i+1=y.\label{const22}
\end{align}
The objective function in \eqref{optimization_auxiliary2} is a convex function. Moreover, the constraints in \eqref{const12} and \eqref{const22} are affine. Hence, Problem \eqref{optimization_auxiliary2} is convex. We use the Lagrangian duality approach to solve Problem  \eqref{optimization_auxiliary2}. Problem \eqref{optimization_auxiliary2} satisfies Slater's conditions. Thus, the Karush-Kuhn-Tucker (KKT) conditions are  both  necessary  and  sufficient for optimality \cite{boyd2004convex}. Let $\mathbf{\gamma}=(\gamma_1, \ldots, \gamma_M)$ and $\mu$ be the Lagrange multipliers associated with constraints \eqref{const12} and \eqref{const22}, respectively. Then, the Lagrangian of Problem \eqref{optimization_auxiliary2} is given by
\begin{equation}\label{lag88}
\begin{split}
\!\!\!\!L(\mathbf{r},\mathbf{\gamma},\mu)=&\sum_{i=1}^M\left[\frac{w_iy}{r_i}+w_i\right]\\&+\sum_{i=1}^M\!\gamma_i(r_i\!-\!b_iy)+\mu\left(\sum_{i=1}^Mr_i\!+\!1\!-\!y\right).\!\!\!\!
\end{split}
\end{equation}
Take the derivative of \eqref{lag88} with respect to $r_l$ and set it equal to 0, we get
\begin{align}
\frac{-w_ly}{r_l^2}+\gamma_l+\mu=0.
\end{align}
This and KKT conditions imply
\begin{align}
&r_l=\sqrt{\frac{w_ly}{\gamma_l+\mu}},\\
&\gamma_l\geq 0, r_l-b_ly\leq 0,\\
&\gamma_l( r_l-b_ly)=0,\\
&\sum_{i=1}^Mr_i+1=y.
\end{align}
If $\gamma_l=0$, then $r_l=\sqrt{(w_ly)/\mu}$ and $r_l\leq b_ly$; otherwise, if $\gamma_l>0$, then $r_l=b_ly$ and $r_l<\sqrt{(w_ly)/\mu}$. Hence, we have
\begin{equation}
r_l=\min\left\lbrace b_ly,\sqrt{\frac{w_ly}{\mu^{\star}}}\right\rbrace,
\end{equation}
where by \eqref{const22}, $\mu^{\star}$ satisfies 
\begin{align}
\sum_{i=1}^M\min\left\lbrace b_iy,\sqrt{\frac{w_iy}{\mu^{\star}}}\right\rbrace+1=y.
\end{align}
We can observe that $\mu^{\star}$ is a function of $y$. Because of that, we can define $\beta^{\star}(y)=\sqrt{1/(y\mu^{\star})}$, which is a function of $y$ as well. Then, the optimum solution to \eqref{optimization_auxiliary2} can be rewritten as
\begin{equation}\label{solution_cond_1_2_2}
r_l=\min\{b_l,\beta^{\star}(y)\sqrt{w_l}\}y, \forall l,
\end{equation}
where $\beta^{\star}(y)$ satisfies
\begin{equation}\label{solution_cond_1_1_1}
\sum_{i=1}^M\min\{b_i,\beta^{\star}(y)\sqrt{w_i}\}+\frac{1}{y}=1.
\end{equation}
Substituting \eqref{solution_cond_1_2_2} and \eqref{solution_cond_1_1_1} back in Problem \eqref{optimization_auxiliary2}, we get the following optimization problem (the outer layer):
\begin{align}
\begin{split}\label{eqa55}
\min_{y>1}& \sum_{i=1}^M\left[\frac{w_i}{\min\{b_i,\beta^{\star}(y)\sqrt{w_i}\}}+w_i\right]
\end{split}\\
\textbf{s.t.}~&\sum_{i=1}^M\min\{b_i,\beta^{\star}(y)\sqrt{w_i}\}+\frac{1}{y}=1.\label{eqa56}
\end{align}
Problem \eqref{eqa55} serves as a lower bound of Problem \textbf{2}, and hence a lower bound of Problem \textbf{1}. We can observe that the objective function in \eqref{eqa55} is decreasing in $\beta^{\star}(y)$. Moreover, \eqref{eqa56} implies that $\beta^{\star}(y)$ is strictly increasing in $y$ if $\sum_{i=1}^Mb_i\geq 1$. As a result, $y=\infty$ is the optimal solution of Problem \eqref{eqa55}. At the limit, the constraint \eqref{eqa56} converges to \eqref{condition_beta}. Since $\beta^{\star}$ serves as a solution for \eqref{condition_beta}, we can deduce that $\lim_{y\to\infty}\beta^\star (y)=\beta^\star$. Thus, we have the following lower bound:
\begin{equation}\label{lower_bound_final_sumbi_ge1}
\bar{\Delta}^{\text{peak}}_{\text{opt}}\ge \bar{\Delta}^{\text{peak}}_{\text{opt},2}\geq \sum_{i=1}^M\left[\frac{w_i}{\min\{b_i,\beta^{\star}\sqrt{w_i}\}}+w_i\right].
\end{equation}
This completes the proof.

\section{Proof of Lemma \ref{lemma_feasibility2}}\label{Appendix_C}
Because $1-e^{-x}\leq x$, we can obtain
\begin{align}
\begin{split}
&r_le^{-r_l\frac{t_s}{\mathbb{E}[T]}}+[1-e^{-r_l\frac{t_s}{\mathbb{E}[T]}}]\sum_{i=1}^Mr_i\\&=r_l+[1-e^{-r_l\frac{t_s}{\mathbb{E}[T]}}]\left(\sum_{i=1}^Mr_i-r_l\right)\\&\leq  r_l+r_l\frac{t_s}{\mathbb{E}[T]}\left(\sum_{i=1}^mr_i-r_l\right),
\end{split}
\end{align}
Hence, if $\mathbf{r}$ satisfies the constraint
\begin{align}\label{tighter_constraint_ac}
\frac{r_l+r_l\frac{t_s}{\mathbb{E}[T]}\left(\sum_{i=1}^Mr_i-r_l\right)}{\sum_{i=1}^Mr_i+1}\le b_l,
\end{align}
then $\mathbf{r}$ also satisfies the constraint of Problem \textbf{1} in \eqref{problem1}. Consider the following set of solution indexed by a parameter $c>0$:
\begin{align}
&r_l=cu_l,~\forall l,\label{proposed_sol_ac}\\
&u_l=\frac{b_l}{1-\sum_{i=1}^Mb_i},~\forall l\label{proposed_sol_ac103}
\end{align}
We want to find a $c$ such that the solution in \eqref{proposed_sol_ac} and \eqref{proposed_sol_ac103} is feasible for Problem \textbf{1}. To achieve this, we first substitute the solution \eqref{proposed_sol_ac} and \eqref{proposed_sol_ac103} into the constraint \eqref{tighter_constraint_ac}, and get
\begin{align}\label{lem4_4eq104}
\frac{cu_l+c^2u_l\frac{t_s}{\mathbb{E}[T]}\left(\sum_{i=1}^Mu_i-u_l\right)}{c\sum_{i=1}^Mu_i+1}\le b_l.
\end{align} 
If equality is satisfied in \eqref{lem4_4eq104}, we can
obtain the following quadratic equation for c: 
\begin{align}\label{lem4_4eq105}
c^2\left[u_l\frac{t_s}{\mathbb{E}[T]}\left(\sum_{i=1}^Mu_i\!-\!u_l\right)\right]\!+\!c\left(u_l\!-\!b_l\sum_{i=1}^Mu_i\right)\!-\!b_l=0.
\end{align}
The solution to \eqref{lem4_4eq105} is given by $c_l$ in \eqref{feasible_factor}. Hence, $r_l = c_l u_l$ is feasible for the constraint \eqref{tighter_constraint_ac} for source~$l$. 

As feasibility for one source only is insufficient, we further prove that the solution in \eqref{proposed_sol_ac} and \eqref{proposed_sol_ac103} with $c=\min_lc_l$ is feasible for satisfying the energy constraints of all sources $l = 1,\ldots, M$. To that end, let us consider the monotonicity of the LHS of \eqref{lem4_4eq104}. By taking the  derivative with respect to $c$, we get
\begin{align}
\frac{u_l\frac{t_s}{\mathbb{E}[T]}\left(\sum_{i=1}^Mu_i-u_l\right)\left(c^2\sum_{i=1}^Mu_i+2c\right)+u_l}{(c\sum_{i=1}^Mu_i+1)^2}>0.
\end{align}
Hence, 
\begin{align}\label{lem4_4eq106}
r_l = \left(\min_{l} c_l\right) u_l, ~\forall l,
\end{align}
is feasible for the energy constraints of all sources $l = 1,\ldots, M$. After some manipulations, the solution in \eqref{proposed_sol_ac103} and \eqref{lem4_4eq106} are equivalently expressed as \eqref{r_f_b_gr_1} and  \eqref{condition_x*_beta*} - \eqref{feas_fact_eq27}. This completes the proof.

\section{Proof of Lemma \ref{lemma_lower_bound_sum_i_b_i_le_1}}\label{Appendix_D}
By replacing $e^{-r_l(t_s/\mathbb{E}[T])}/r_l$ by $e^{-\sum_{i=1}^Mr_i(t_s/\mathbb{E}[T])}/[$ $b_l(\sum_{i=1}^Mr_i+1)]$  and $e^{\sum_{i=1}^Mr_i(t_s/\mathbb{E}[T])}$ by 1 in \eqref{problem2} of Problem \textbf{2}, we obtain the following optimization problem:
\begin{align}
\begin{split}\label{lower_bound_case_bf_le_1_1}
\min_{r_l>0}& \sum_{l=1}^M\frac{w_le^{-\sum_{i=1}^Mr_i\frac{t_s}{\mathbb{E}[T]}}}{b_i}+\sum_{l=1}^Mw_l
\\
\textbf{s.t.}~&r_l\leq b_l\left(\sum_{i=1}^Mr_i+1\right), \forall l.
\end{split}
\end{align}
Since  $r_l\leq b_l(\sum_{i=1}^Mr_i$ $+1)$, we have
\begin{align}
 \frac{e^{-r_l\frac{t_s}{\mathbb{E}[T]}}}{r_l}\geq \frac{e^{-\sum_{i=1}^Mr_i\frac{t_s}{\mathbb{E}[T]}}}{b_l\left(\sum_{i=1}^Mr_i+1\right)}.
 \end{align}
 Moreover, we have  $e^{\sum_{i=1}^Mr_i(t_s/\mathbb{E}[T])}\geq 1$. Thus, Problem \eqref{lower_bound_case_bf_le_1_1} serves as a lower bound of Problem \textbf{2}, and hence a lower bound of Problem \textbf{1} as well. By removing the constant term $\sum_{l=1}^Mw_l$ in the objective function of Problem  \eqref{lower_bound_case_bf_le_1_1} and then taking the logarithm, Problem \eqref{lower_bound_case_bf_le_1_1} is reformulated as 
 \begin{align}
\begin{split}\label{problem64}
\min_{r_i>0}& \log\left(\sum_{i=1}^M\frac{w_i}{b_i}\right)-\sum_{i=1}^Mr_i\frac{t_s}{\mathbb{E}[T]}
\\
\textbf{s.t.}~&r_l\leq b_l\left(\sum_{i=1}^Mr_i+1\right), \forall l.
\end{split}
\end{align}
Obviously, Problem \eqref{problem64} is a convex optimization problem and satisfies Slater's  conditions. Thus,  the KKT conditions are are necessary and sufficient for optimality. Let $\mathbf{\tau}=(\tau_1,\ldots,\tau_M)$ be the Lagrange multipliers associated with the  constraints of Problem  \eqref{problem64}. Then, the Lagrangian of Problem \eqref{problem64} is given by 
\begin{equation}\label{lag110}
\begin{split}
L(\mathbf{r}, \mathbf{\tau})=&\log\left(\sum_{i=1}^M\frac{w_i}{b_i}\right)-\left(\sum_{i=1}^Mr_i\frac{t_s}{\mathbb{E}[T]}\right)\\&+\sum_{i=1}^M\tau_i\left[r_i-b_i\left(\sum_{i=1}^Mr_i+1\right)\right].
\end{split}
\end{equation}
Take the derivative of \eqref{lag110} with respect to $r_l$ and set it equal to 0, we get
\begin{align}
\frac{-t_s}{\mathbb{E}[T]}+\tau_l(1-b_l)-\sum_{i\neq l}\tau_ib_i=0.
\end{align}
This and KKT conditions imply
\begin{align}
&\tau_l=\frac{t_s}{\mathbb{E}[T](1-b_l)}+\frac{\sum_{i\neq l}\tau_ib_i}{1-b_l},\label{eqa67}\\&
\tau_l\geq 0, r_l-b_l\left(\sum_{i=1}^Mr_i+1\right)\leq 0,\label{eqa68}\\&
\tau_l\left[r_l-b_l\left(\sum_{i=1}^Mr_i+1\right)\right]=0.\label{eqa69}
\end{align}
Since $\sum_{i=1}^Mb_i<1$, \eqref{eqa67} implies that $\tau_l>0$ for all $l$. This and \eqref{eqa69} result in 
\begin{equation}\label{eqa70}
r_l=b_l\left(\sum_{i=1}^Mr_i+1\right), \forall l.
\end{equation}
Because $\sum_{i=1}^Mb_i<1$, \eqref{eqa70} has a unique solution, which is given by 
\begin{equation}\label{eqa71}
r_l=\frac{b_l}{1-\sum_{i=1}^Mb_i}, \forall l.
\end{equation}
Hence, the solution to \eqref{lower_bound_case_bf_le_1_1} and \eqref{problem64} is given by \eqref{eqa71}. Substitute \eqref{eqa71} into \eqref{lower_bound_case_bf_le_1_1}, we get the following lower bound:
\begin{equation}\label{lower_bound_72}
\bar{\Delta}^{\text{peak}}_{\text{opt}}\ge \bar{\Delta}^{\text{peak}}_{\text{opt},2}\geq \sum_{l=1}^M\frac{w_le^{\frac{-\sum_{i=1}^Mb_i}{1-\sum_{i=1}^Mb_i}\frac{t_s}{\mathbb{E}[T]}}}{b_l}+\sum_{l=1}^Mw_l.
\end{equation}
This completes the proof.

\section{Proof of Corollary \ref{corollary_centeralized_scheduler}}\label{Appendix_E}
We start by solving Problem \eqref{problem_centr} for optimal $\mathbf{a}$.  Problem \eqref{problem_centr} is a convex optimization problem and satisfies Slater's conditions. Thus, the KKT conditions are necessary and sufficient for optimality. Let $\mathbf{\lambda}=(\lambda_1,\ldots,\lambda_M)$ and $\nu$ be the Lagrange multipliers associated with the constraints \eqref{eq105} and \eqref{eq106}, respectively. Then, the Lagrangian of Problem \eqref{problem_centr} is given by
\begin{align}\label{lag119}
\begin{split}
L(\mathbf{a},\mathbf{\lambda},\nu)=&\sum_{i=1}^M\left[\frac{w_i}{a_i}+w_i\right]\\&+\sum_{i=1}^M\lambda_i(a_i-b_i)+\nu\left(\sum_{i=1}^Ma_i-1\right).
\end{split}
\end{align}
Take the derivative of \eqref{lag119} with respect to $a_l$ and set it equal to 0, we get
\begin{align}
\frac{-w_l}{a_l^2}+\lambda_l+\nu=0.
\end{align}
This and KKT conditions imply
\begin{align}
&a_l=\sqrt{\frac{w_l}{\lambda_l+\nu}},\\&
\lambda_l\ge 0,~a_l-b_l\le 0,\\&
\lambda_l(a_l-b_l)=0,\\&
\sum_{i=1}^Ma_i=1.\label{eq110}
\end{align}
In case $\lambda_l=0$, then we have that $a_l=\sqrt{w_l/\nu}$ and $a_l-b_l\le 0$; otherwise, if $\lambda_l>0$, then we have $a_l=b_l$ and $a_l\le \sqrt{w_l/\nu}$. Thus, the optimal solution is given by
\begin{align}
a_l^{\star}=\min\left\lbrace b_l,\sqrt{\frac{w_l}{\nu^{\star}}}\right\rbrace,
\end{align}
where by \eqref{eq110}, $\nu^{\star}$ satisfies
\begin{align}\label{eq126}
\sum_{i=1}^M\min\left\lbrace b_i,\sqrt{\frac{w_i}{\nu^{\star}}}\right\rbrace=1.
\end{align}
By comparing \eqref{eq126} with \eqref{condition_beta}, we can deduce that $\sqrt{1/\nu^{\star}}=\beta^{\star}$, where $\beta^{\star}$ satisfies
\begin{align}\label{eq111}
\sum_{i=1}^M\min\{b_i,\beta^{\star}\sqrt{w_i}\}=1.
\end{align}
Since $\sum_{i=1}^Mb_i\ge 1$, \eqref{eq111} has a solution for $\beta^{\star}$ as shown in Lemma \ref{lemma_solution_19}. Hence, the solution to Problem  \eqref{problem_centr} can be rewritten as
\begin{align}\label{sol127}
a_l=\min\{b_l,\beta^{\star}\sqrt{w_l}\},~\forall l.
\end{align}
Substituting \eqref{sol127} into \eqref{problem_centr}, we obtain
\begin{align}
\bar{\Delta}_{\text{opt-s}}^{\text{peak}}=\sum_{i=1}^M\left[\frac{w_i}{\min\{b_i,\beta^{\star}\sqrt{w_i}\}}+w_i\right],
\end{align}
which is equal to the asymptotic optimal value of Problem \textbf{1} in \eqref{asymptotic_value_final}. This completes the proof. 
\fi

\end{document}